\documentclass[11pt,letterpaper,onecolumn,accepted=2024-11-01]{quantumarticle}%
\pdfoutput=1 

\usepackage[utf8]{inputenc}
\usepackage[english]{babel}
\usepackage[T1]{fontenc}

\usepackage{latexsym,amssymb,amsfonts,amsmath,mathrsfs,hyperref,mathtools}
\usepackage{amsthm}

\usepackage{pifont}%
\usepackage{enumerate}
\hypersetup{colorlinks,citecolor=NavyBlue,filecolor=black,linkcolor=black,urlcolor=black}
\usepackage[dvipsnames]{xcolor}
\usepackage{comment} 
\usepackage[capitalize]{cleveref}

\newcommand{\vertiii}[1]{{\left\vert\kern-0.25ex\left\vert\kern-0.25ex\left\vert #1 
		\right\vert\kern-0.25ex\right\vert\kern-0.25ex\right\vert}}

  \renewcommand{\Pr}{\mbox{\rm Pr}}
  \newcommand{\Exp}{{\mathbb{E}}}

  \DeclareMathOperator{\cb}{cb}
  \renewcommand{\cb}{{\mathrm{cb}}}
     
  \DeclareMathOperator{\jcb}{jcb}
  
  \newcommand{\norm}[1]{\|#1\|}

  \newcommand{\R}{\mathbb{R}} 
  \newcommand{\C}{\mathbb{C}} 
  \newcommand{\N}{\mathbb{N}} 
  \newcommand{\Z}{\mathbb{Z}} 
  \newcommand{\pmset}[1]{\{-1,1\}^{#1}} 
  \newcommand{\bset}[1]{\{0,1\}^{#1}} 

  \newcommand{\st}{\,\mid\,} 
    
  \newcommand{\eg}{{e.g.}} 
  \newcommand{\eps}{\varepsilon}

  \newcommand{\id}{I}
  \newcommand{\A}{\mathcal A}

  \DeclareMathOperator{\supp}{supp}

  \newcommand{\err}{\mathcal E}
  
  \newcommand{\ind}[1]{\mathbf{#1}}
   
  \newcommand{\beq}{\begin{equation}}
  \newcommand{\eeq}{\end{equation}}
  \newcommand{\beqn}{\begin{equation*}}
  \newcommand{\eeqn}{\end{equation*}}
  \newcommand{\beqr}{\begin{eqnarray}}
  \newcommand{\eeqr}{\end{eqnarray}}
  \newcommand{\beqrn}{\begin{eqnarray*}}
  \newcommand{\eeqrn}{\end{eqnarray*}}
  \newcommand{\bmline}{\begin{multline}}
  \newcommand{\emline}{\end{multline}}
  \newcommand{\bmlinen}{\begin{multline*}}
  \newcommand{\emlinen}{\end{multline*}}

  \theoremstyle{plain}
  \newtheorem{theorem}{Theorem}[section]
  \newtheorem{lemma}[theorem]{Lemma}
  \newtheorem{proposition}[theorem]{Proposition}
  
  \newtheorem{corollary}[theorem]{Corollary}
  \newtheorem{question}[theorem]{Question}
  \theoremstyle{definition}
  \newtheorem{definition}[theorem]{Definition}
  \newtheorem{example}[theorem]{Example}

  \theoremstyle{remark}
  \newtheorem{remark}[theorem]{Remark}
  
  \renewenvironment{proof}[1][]{
    	\begin{trivlist}
     	\item[\hspace{\labelsep}{\em\noindent Proof#1:\/}]}
     	{{\hfill$\Box$}
    	\end{trivlist}
  }
  
\newif\ifnotes\notesfalse

\ifnotes
\usepackage{color}
\definecolor{mygrey}{gray}{0.50}
\newcommand{\notename}[2]{{\textcolor{cyan}{\footnotesize{\bf (#1:} {#2}{\bf ) }}}}
\newcommand{\noteswarning}{{\begin{center} {\Large WARNING: NOTES ON}\end{center}}}

\else

\newcommand{\notename}[2]{{}}
\newcommand{\noteswarning}{{}}

\fi

\usepackage{thm-restate} 

\begin{document}

\title[No converse for the polynomial method]{Grothendieck inequalities characterize converses to the polynomial method
}

\author[J. Bri\"{e}t]{Jop Bri\"{e}t}
\address{CWI \& QuSoft, Science Park 123, 1098 XG Amsterdam, The Netherlands}
\email{j.briet@cwi.nl}

\author[F. Escudero Guti\'errez]{Francisco Escudero Guti\'errez}
\address{CWI \& QuSoft, Science Park 123, 1098 XG Amsterdam, The Netherlands}
\email{feg@cwi.nl}

\author[S. Gribling]{Sander Gribling}
\address{Tilburg University, Warandelaan 2,  5037 AB Tilburg, The Netherlands}
\email{s.j.gribling@tilburguniversity.edu}

\thanks{\vspace{0.3cm}\\ This version extends arXiv version~1 of this work with most of~\cite{briet2022converses}, which appeared in the proceedings of TQC'22. \includegraphics[height=2ex]{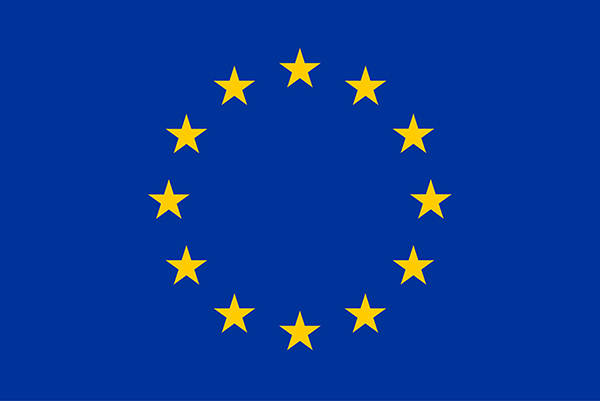} This research was supported by the European Union’s Horizon 2020 research and innovation programme under the Marie Sk{\l}odowska-Curie grant agreement no. 945045, and by the NWO Gravitation project NETWORKS under grant no. 024.002.003.}

\begin{abstract}
	A surprising `converse to the polynomial method' of Aaronson et al.\ (CCC'16) shows that any bounded quadratic polynomial can be computed exactly in expectation by a 1-query algorithm up to a universal multiplicative factor related to the famous Grothendieck constant. Here we show that such a result does not generalize to quartic polynomials and $2$-query algorithms, even when we allow for additive approximations. 
	We also show that the additive approximation implied by their result is tight for bounded bilinear forms, which gives a new characterization of the Grothendieck constant in terms of $1$-query quantum algorithms.  
Along the way we provide reformulations of the completely bounded norm of a form, and its dual norm.
\end{abstract}

\maketitle

\noteswarning

\section{Introduction}\label{sec:intro}
Quantum query complexity is one of the few models of computation in which the strengths and weaknesses of quantum computers can be rigorously studied with currently-available techniques (see \eg,~\cite{Ambainis:2018, Aaronson:2021} for recent surveys).
On the one hand, many of quantum computing's best-known algorithms, such as for unstructured search~\cite{Grover:1996}, period finding (the core of Shor's algorithm for integer factoring)~\cite{Shor:1997} and element distinctness~\cite{Ambainis:2007}, are most naturally described in the query model.
On the other hand, the model admits powerful lower-bound techniques.

For a (possibly partial) Boolean function $f:D\to\pmset{}$ defined on a set $D\subseteq \pmset{n}$, the celebrated \emph{polynomial method} of Beals, Buhrman, Cleve, Mosca and de~Wolf~\cite{polynomialmethod} gives a lower bound on the quantum query complexity of~$f$, denoted~$Q(f)$, in terms of the minimal degree of an approximating polynomial for~$f$, or approximate degree, $\widetilde{\deg}(f)$.
The method relies on the basic fact that for any $t$-query quantum algorithm~$\mathcal{A}$ that takes an $n$-bit input and returns a sign, there is a real $n$-variable polynomial~$p$ of degree at most~$2t$ such that $p(x)=\mathbb{E}[\mathcal{A}(x)]$ for every~$x$. Here, the expectation is taken with respect to the randomness in the measurement done by~$\mathcal A$.\footnote{We identify quantum query algorithms with the (random) functions giving their outputs.}
Using this method, many well-known quantum algorithms were proved to be optimal in terms of query complexity (see \eg,~\cite{BunKhotariThaler:2020} and references therein).

Since polynomials are simpler objects than quantum query algorithms, it is of interest to know how well approximate degree approximates quantum query complexity.
There are total functions~$f$ that satisfy $Q(f) \geq \widetilde{\deg}(f)^{c}$ for some absolute constant $c>1$~\cite{Ambainis:2006, AaronsonBDK:2016};
the second reference gives an exponent $c = 4 - o(1)$, which was shown to be optimal in~\cite{AaronsonBDK:2016}. For partial functions it was recently shown that this separation can even be exponential~\cite{ambainis2023exponentialccc}. These separations rule out a direct converse to the polynomial method, whereby a given bounded degree-$2t$ polynomial~$p$ can be computed by a $t$-query quantum algorithm~$\mathcal A$.
However, since these results concern functions whose approximate degree grows with~$n$, they leave room for the possibility that such an~$\mathcal A$ approximates~$p$ with some error that depends on~$t$.

We will say that a polynomial~$p$ is \emph{bounded} if its restriction to the Boolean hypercube takes values in the interval~$[-1,1]$ and that~$\mathcal A$ \emph{approximates}~$p$ if for some constant \emph{additive} error parameter $\eps < 1$, we have that $|p(x) - \Exp[\mathcal A(x)]| \leq \eps$ for every~$x$. 
Note that an additive error of~1 can trivially be achieved with a uniformly random coin flip.
For a function $p:\{-1,1\}^n\to\mathbb{R}$ and positive integer~$t$, we denote the smallest additive error that a $t$-query quantum algorithm can achieve by
\begin{align}
\err(p ,t) := \inf\big\{ \eps \geq 0 \st &\ \exists \ t\text{-query quantum algorithm } \A \text{ with } \label{eq:Eft}\\
&\ |p(x) - \Exp[\mathcal A(x)]| \leq \eps \quad \forall x \in \pmset{n}\big\}. \nonumber
\end{align}

For bounded polynomials of degree at most~2, a ``multiplicative converse'' to the polynomial method was proved in~ \cite{Aaronson2015PolynomialsQQ}, showing that up to an absolute constant scaling, quadratic polynomials can indeed be computed by 1-query quantum algorithms.

\begin{theorem}[Quadratic multiplicative converse~\cite{Aaronson2015PolynomialsQQ}]\label{query1}
	There exists an absolute constant $C\in(0,1]$ such that  $\mathcal{E}(Cp,1)=0$ for every bounded polynomial~$p$ of degree at most~$2$.
\end{theorem}

This result directly implies the following additive version.

\begin{corollary}[Quadratic additive converse] \label{cor:addconversedeg2}
	There exists an absolute constant $\eps\in(0,1)$ such that the following
	holds. For every bounded polynomial~$p$ of degree at most~$2$, we have $\mathcal{E}(p,1)\leq \eps$.
	In particular, one can take $\eps = 1 - C$ for the constant~$C$ appearing in \cref{query1}.
\end{corollary}

In light of the polynomial method, \cref{cor:addconversedeg2} shows that one-query quantum algorithms are roughly equivalent to bounded quadratic polynomials. 
The authors of~\cite{Aaronson2015PolynomialsQQ} asked whether this result generalizes to higher degrees. 
Two ways to interpret this question are that for any~$k$, any bounded degree-$2k$ polynomial~$p$ satisfies:
\begin{enumerate}[(a)]
\setlength{\leftskip}{2em}
 \item  Multiplicative converse: $\err(Cp,k) =0$ for some $C = C(k)>0$, or;\label{item:optiona}
\item Additive converse:  $\err(p, k)\leq \eps$ for some $\eps = \eps(k)<1$.\label{item:optionb}
\end{enumerate}
The dependence on the degree~$k$ in these options is necessary due to the known separations between bounded-error quantum query complexity and approximate degree. Option~\eqref{item:optiona}, the higher-degree version of Theorem~\ref{query1}, was ruled out  in 
\cite{QQA=CBF}.

\begin{restatable}{theorem}{NoMultConverse}\label{theo:NoMultConverse}
	For any $C>0$, there exist an $n\in \N$ and a bounded quartic $n$-variable polynomial~$p$ such that no two-query quantum algorithm~$\mathcal A$  satisfies $\Exp[\mathcal{A}(x)] = C p(x)$ for every~$x\in\{-1,1\}^n$.
\end{restatable}

Note that Option~\eqref{item:optiona} with $C$ implies Option~\eqref{item:optionb} with $1-C$, but \cref{theo:NoMultConverse} does not rule out Option~\eqref{item:optiona}.

\subsection{Our contributions.}


Our first contribution concerns an error in the original proof of \cref{theo:NoMultConverse}, which was based on a probabilistic example. 
Here, we show that Theorem~\ref{theo:NoMultConverse} holds as stated, both by considering a slightly modified probabilistic example and by giving a completely explicit example.
More importantly, we prove a stronger result that subsumes \cref{theo:NoMultConverse}: we rule out the possibility of Option~\eqref{item:optionb}.

\begin{restatable}{theorem}{NoAdditiveConverse}\label{theo:NoAdditiveConverse}
	There is no constant $\eps\in (0,1)$ such that for every bounded polynomial $p$ of degree at most~$4$, we have $\mathcal{E} (p,2)\leq \eps.$
\end{restatable}

In the context of quantum query complexity of Boolean functions, this rules out arguably the most natural way to \emph{upper} bound~$Q(f)$ in terms of $\widetilde{\deg}(f)$: First, $\eps$-approximate $f$ by a degree-$2t$ polynomial $p$, then $\eps'$-approximate $p$ with a $t$-query quantum algorithm $\mathcal A$, 
with $\eps+\eps'<1$, and finally boost the success probability of $\mathcal A$ so that it approximates $f$, for instance by taking the majority of independent runs of~$\mathcal A$. Corollary~\ref{cor:addconversedeg2} gives the only exceptional case where this is possible in general.


Our second contribution concerns 1-query quantum algorithms. \cref{query1} was proved using a surprising application of Grothendieck's theorem from Banach space theory~\cite{grothendieck1953resume} (see also \cref{sec:intro}).
For bounded bilinear forms, $p(x,y) = x^{\mathsf T}Ay$ given by a matrix $A\in \R^{n\times n}$, the result holds with~$1/C$ equal to the Grothendieck constant~$K_G$ (see~\cite[Section~5]{QQA=CBF} for a short proof).
Determining the precise value of~$K_G$ is a notorious open problem posed in~\cite{grothendieck1953resume}; the best-known lower and upper bounds place it in the interval $(1.676, 1.782)$~\cite{Davie:1984, Reeds:1991, Braverman:2013}.
The general form of \cref{query1} follows from decoupling techniques.
It is not hard to show that~$1/K_G$ is the optimal constant in the bilinear case for the multiplicative setting of \cref{query1}.
Here, we also show that the additive  approximation implied by the multiplicative setting is optimal.

\begin{restatable}{theorem}{characterizingKG}\label{theo:characterizingKG}
	The worst-case minimum error for one-query quantum algorithms satisfies 
 $$\sup_{p}\mathcal{E}(p,1)=1-\frac{1}{K_G},$$ where the supremum is taken over the set of bounded bilinear forms.

\end{restatable}

This complements another well-known characterization of~$K_G$ in terms of the largest-possible Bell-inequality violations in two-player XOR games~\cite{Tsirelson}.

\subsection{Our main technical result}

Both \cref{theo:NoMultConverse,theo:NoAdditiveConverse} are in fact corollaries of our main result (\cref{theo:lowerSDP} below), which gives a formula for $\err(p,t)$ when $p$ is a block-multilinear form. 
Block-multilinear forms already played an important role in other works related to quantum query complexity~\cite{o2015polynomial,Aaronson2015PolynomialsQQ, Bansal:2022}, theoretical computer science \cite{khot2007linear,lovett2010elementary,kane2013prg} and in the polarization theory of functional analysis \cite{bohnenblust1931absolute,harris1972bounds}. 

The formula characterizes $\err(p,t)$ in terms of a ratio of norms appearing naturally in Grothendieck's theorem for bilinear forms. A polynomial $p$ is a bilinear form if it can be expressed as $p(x,y) = x^{\mathsf T}Ay$ for a matrix~$A\in \R^{n\times n}$.
A central role is played by two norms associated to~$A$ (and hence $p$). 
The first is the $\ell_\infty$-norm of~$A$ when identified with the map $(x,y)\mapsto p(x,y)$,
\begin{equation*}
\|A\|_{\infty} = \max_{x,y\in \{-1,1\}^n} x^{\mathsf T}Ay.
\end{equation*}
The second is its completely bounded norm, which may be given by 
\begin{equation} \label{eq:cbnormdefintro}
\|A\|_{\cb} = \sup_{\substack{d\in\mathbb{N}, u,v:[n]\to S^{d-1}}} \sum_{i,j=1}^n A_{i,j} \langle u(i), v(j)\rangle,
\end{equation}
where~$S^{d-1}$ denotes the Euclidean unit sphere of $\mathbb 
R^d$. 
The celebrated Grothendieck theorem asserts that these norms are equivalent up to a constant factor.

\begin{theorem}[Grothendieck's theorem]\label{thm:GT}
 There exists a constant~$K<\infty$ such that for any $n\in \N$ and $A\in \R^{n\times n}$, we have
 \begin{equation}\label{eq:GT}
 \|A\|_{\infty} \leq \|A\|_{\cb} \leq K\|A\|_{\infty}.
 \end{equation}
\end{theorem}

The non-trivial part of this result is the second inequality in~\eqref{eq:GT} and the smallest~$K$ for which it holds is the above-mentioned Grothendieck constant~$K_G$. 
 Our characterization of $\err(p,1)$ involves the \emph{dual norms} of the $\ell_\infty$-bilinear norm and the completely bounded norm (c.f.\ \cref{sec:notation}). 
 The dual formulation of (the second inequality in) Grothendieck's theorem asserts that, for any matrix~$A$, 
\begin{equation}\label{eq:GT_dual}
\|A\|_{\infty,*} \leq K_G\|A\|_{\cb,*}.
\end{equation}
Similar norms can be defined for block-multilinear forms of higher degree.
Endowing the space of polynomials with the standard inner product of the coefficient vectors in the monomial basis, our formula for $\err(p,t)$ is as follows.  

\begin{theorem}[Informal version of \cref{theo:lowerSDP}] \label{theo:bilinearintro}
    For a block-multilinear form~$p$ of degree $2t$, we have 
    \[
\err(p,t) = \sup_{q}\frac{ \langle p,q\rangle - \|q\|_{\cb,*}}{\|q\|_{\infty,*}}.
\]
where the supremum runs over all block-multilinear forms $q$ of degree $2t$.
\end{theorem}

The proof of \cref{theo:bilinearintro} uses a characterization of quantum query algorithms in terms of completely bounded polynomials~\cite{QQA=CBF}. 
Recently, this characterization was also used to make progress on the problem to determine ``the need for structure in quantum speed-ups'' in works by Bansal, Sinha and de~Wolf \cite{Bansal:2022} and by the second author \cite{gutierrez2023influences}. In addition, it led to a new exact SDP-based formulation for quantum query complexity, due to Laurent and the third author~\cite{gribling2019semidefinite}.
\medskip

\cref{theo:characterizingKG,theo:NoAdditiveConverse} follow from \cref{theo:bilinearintro} by taking suprema over particular sequences of bounded degree-$2t$ block-multilinear forms.
From \cref{theo:bilinearintro} it follows that
\begin{align}\label{eq:superrpt}
    \sup_p\err(p,t) &= \sup_{q}\left[\left(\sup_p\frac{\langle p,q\rangle}{\|q\|_{\infty,*}}\right) -\frac{  \|q\|_{\cb,*}}{\|q\|_{\infty,*}}\right]=1-\inf_{q} \frac{  \|q\|_{\cb,*}}{\|q\|_{\infty,*}}.
\end{align}
Now, \cref{theo:characterizingKG} follows from \cref{eq:superrpt} and 
the dual version of Grothendieck's inequality (\cref{eq:GT_dual}). Similarly, \cref{theo:NoAdditiveConverse} is proven by using \cref{eq:superrpt} and constructing a family of degree-4 polynomials $(p_n)_n$ that witnesses the failure of Grothendieck inequality. By this we mean that $(p_n)_n$ exhibit the separation 
\begin{equation}\label{eq:separationsIntro}
    \frac{\norm{p_n}_{\cb}}{\norm{p_n}_{\infty}}\to \infty.
\end{equation}
By duality this implies that there is a sequence $(r_n)_n$ with $\norm{r_n}_{\cb,*}/\norm{r_n}_{\infty,*}\to 0,$ which alongside \cref{eq:superrpt} implies that $\sup_p\mathcal{E}(p,2)=1,$ as desired.

\subsection{Organization} The rest of the paper is organized as follows. \cref{sec:prelim} compiles the necessary preliminaries, in \cref{sec:maintechresult} we prove our main technical result, in \cref{sec:separations} we show the norm separations displayed in \cref{eq:separationsIntro}, and in \cref{sec:noconverses} we prove \cref{theo:characterizingKG,theo:NoAdditiveConverse}. In \cref{sec:question} we pose an open question, whose positive answer would strenghten our main result and establish a clean link between quantum query algorithms and XOR games. Finally, in \cref{sec:dualnorms} we phrase the dual norms as (efficiently solvable) convex optimization programs and use these formulations to give an alternative proof of our main result. (Strictly speaking, \cref{sec:dualnorms} is not necessary to understand the proofs of \cref{theo:characterizingKG,theo:bilinearintro,theo:NoAdditiveConverse}.)

\section{Preliminaries}\label{sec:prelim}

\subsection{Notation}\label{sec:notation}
For $n\in\mathbb{N}$, write $[n]:=\{1,\dots,n\}$. 
We endow~$\R^d$ with the standard inner product $\langle x, y \rangle = \sum_{i \in [d]} x_i y_i$ and write the resulting norm as~$\|x\| = \sqrt{\langle x,x\rangle}$. 
Denote the set of unit vectors of $\mathbb{R}^{d}$ by $S^{d-1}$. 

We endow the space of matrices~$\R^{d \times d}$ with the standard operator norm. 
We write $M(d)=\R^{d \times d}$ and let $B_{M(d)}$ denote the unit ball in $M(d)$ with respect to the operator norm (i.e., the set of contractions). 

Given a normed vector space $(V, \|\:\:\|)$ with $V\subseteq \R^d$, the dual norm of an element $v\in V$ is given by
\begin{equation*}
    \|v\|_{*} = \sup\{|\langle v,w\rangle| \mid w \in V, \ \|w\| \leq 1\}.
\end{equation*}

\subsection{Polynomials}\label{subsec:normsofpolynomials}

As usual we let $\R[x_1,\dots,x_n]$ be the ring of $n$-variate polynomials with real coefficients, whose elements we write as
\beq\label{eq:poly}
p(x) = \sum_{\alpha\in \Z_{\geq 0}^n} c_\alpha x^\alpha,
\eeq
where $x^\alpha = x_1^{\alpha_1}\cdots x_n^{\alpha_n}$ and $c_\alpha\in \mathbb{R}$. 
We define the support of~$p$ by 
\begin{equation*}
\supp(p) = \{\alpha \in \Z_{\geq 0}^n \st c_\alpha\ne 0\}.    
\end{equation*}
For $\alpha\in \Z_{\geq 0}^n$, write $|\alpha| = \alpha_1 + \cdots + \alpha_n$, which is the degree of the monomial~$x^\alpha$.
A form of degree~$d$ is a homogeneous polynomial of degree~$d$, i.e., a polynomial whose support consists of $\alpha$ for which $|\alpha|=d$. 
Denote by $\R[x_1,\dots,x_n]_{=d}$ the space of forms of degree~$d$.
For~$p$ as in~\cref{eq:poly}, define its homogeneous degree-$d$ part by
\beqn
p_{=d}(x) = \sum_{|\alpha| = d} c_\alpha x^\alpha.
\eeqn
We endow $\mathbb{R}[x_1,\dots,x_n]$ with the inner product given by 
\begin{equation*}
\langle p,q\rangle=\sum_{\alpha\in \mathbb{Z}_{\geq 0}^n} c_{\alpha}c_\alpha',    
\end{equation*}
where $c_{\alpha}$ and $c_{\alpha}'$ are the coefficients of $p$ and $q$, respectively.

Multilinear polynomials (those that are affine in every variable) will play an important role in this work, as they can be identified with the real-valued functions defined on the Boolean hypercube.

We recall the definition of $\norm{\cdot}_{1}$ and $\norm{\cdot}_\infty$, which are seminorms of polynomials in $\mathbb{R}[x_1,\dots,x_n]$, and norms on the space of multilinear polynomials. 
\begin{align*}
	\norm{p}_\infty:=\sup_{x\in\{-1,1\}^n} |p(x)|,\\
	\norm{p}_1:=\mathbb{E}_{x\in\{-1,1\}^n} |p(x)|,
\end{align*}
where the expectation is taken with respect to the uniform probability measure.  

\subsection{Tensors, norms and quantum query complexity}

Toward defining the completely bounded norm, we restrict our attention to forms only. 
For every $p\in\mathbb{R}[x_1,\dots,x_n]_{=t}$ there are many $t$-tensors $T \in\mathbb{R}^{n\times \dots\times n}$ such that $T(x)=p(x)$ for every $x\in \mathbb{R}^n$, where $$T(x):=\sum_{\ind{i}\in [n]^t}T_\ind{i}x(\ind{i}),$$ 
where $x(\ind{i}):=x_{i_1}x_{i_2} \cdots x_{i_t}$. 
We let $\mathcal{S}_n$ be the symmetric group on~$n$ elements.
These tensors only have to satisfy 
\begin{equation}\label{eq:fromTensorsToCoefficients}
	\sum_{\ind{i}\in \mathcal{I}_\alpha}T_{\ind{i}}=c_\alpha\ \ \  \forall\alpha\in\mathbb{Z}_{\geq 0}^n,
\end{equation} where $\mathcal{I}_\alpha$ is the set of elements $\ind{i}$ of $[n]^t$ such that the element $m\in [n]$ occurs~$\alpha_m$ times in $\ind{i}$.
Each of these tensors gives a way of evaluating~$p$ in matrices, namely for every matrix-valued map $A:[n]\to M(d)$, $$T(A):=\sum_{\ind{i}\in [n]^t}T_\ind{i}A(\ind{i}),$$ where $A(\ind{i}) :=A(i_1) A(i_2)\dots A(i_t)$.   It is also important to consider the unique symmetric $t$-tensor $T_p\in \mathbb{R}^{n\times\dots\times n}$  such that $$p(x)=T_p(x).$$ 
For a tensor $T = (T_{\ind i})_{\ind{i} \in [n]^t}$ and a permutation $\sigma\in\mathcal{S}_t$ we define $T \circ \sigma := (T_{\sigma(\ind{i})})_{\ind{i} \in [n]^t}$, where $\sigma(\ind i):=((\sigma(i_1),\dots,\sigma(i_d))$. We say that $T$ symmetric if $T = T \circ \sigma$ for all $\sigma \in \mathcal S_t$. The entries of this $T_p$ are given by
\begin{equation}\label{eq:fromCjtoTp}
	(T_p)_\ind{i}=\frac{c_{e_{i_1}+\dots+e_{i_t}}}{\tau(i_1,\dots,i_t)}\ \mathrm{for}\ \ind{i}\in [n]^t,
\end{equation}
where $\tau(i_1,\dots,i_t)$ is the number of distinct permutations of the sequence $(i_1,\dots,i_t)$ and $\{e_i\}$ is an orthonormal basis of $\mathbb R^n$. The completely bounded norm of a tensor $T$ is given by\footnote{This is the completely bounded norm of $T$ when regarded as an element of $\ell^1_n\otimes_h\dots\otimes_h\ell^1_n$, where $h$ stands for the Haagerup tensor product, which determines a well-studied tensor norm. See for instance \cite[Chapter 17]{paulsenoperatoralgebras}.}
\begin{align}
\norm{T}_{\cb} = \sup  &\Big\{\Big\|\sum_{\ind{i}\in [n]^t}T_{\ind{i}}A_1(i_1)\dots A_t(i_t)\Big\| \st  A_s\colon[n] \to B_{M(d)},\ s \in [t],\ d\in\mathbb{N}\Big\}. \label{eq:defTcb}
\end{align}
It is not hard to show that for a $2$-tensor (matrix)~$T$ this definition agrees with the one given in the introduction in \cref{eq:cbnormdefintro}.

The completely bounded norm of a form $p$ is given by
\begin{equation}\label{eq:cboriginal}
	\norm{p}_{\cb}=\inf \sum_{\sigma\in\mathcal{S}_t}\norm{T^\sigma}_{\cb},
\end{equation}
where the infimum runs over all the families of $t$-tensors $\{T_\sigma\}_{\sigma\in S_t}$ such that $\sum_{\sigma\in\mathcal{S}_t} T^\sigma\circ\sigma=T_p.$ In \cref{sec:reformulationcbnorm} we provide easier expressions for both $\norm{T}_{\cb}$ (showing that a single contraction-valued map $A\colon[n] \to B_{M(d)}$ suffices) and $\norm{p}_{\cb}$. 
To prove our main technical result, Theorem~\ref{theo:lowerSDP}, which expresses~$\mathcal E(p,t)$ in the language of completely bounded polynomials, we combine these simplified expressions with the following characterization of quantum query algorithms, proved in~\cite{QQA=CBF}.

\begin{theorem}[Quantum query algorithms are completely bounded forms]\label{theo:generalexpressionforE}
	Let $p:\{-1,1\}^n\to[-1,1]$  and let $t\in\mathbb{N}$. Then,
	\begin{align*}
		\mathcal{E}(p ,t)=\inf \quad &\norm{p-q}_{\infty}\\
						 \mathrm{s.t.} \ \ \, &h \in \R[x_1,\ldots,x_{n+1}]_{=2t} \text{ with } 
						 \norm{h}_{\cb}\leq 1\\
				& q:\{-1,1\}^n\to\mathbb{R},\ \mathrm{with}\ q(x)=h(x,1)\ \forall\ x\in\{-1,1\}^n.
	\end{align*}
\end{theorem}

\subsection{Block-multilinear forms} Our main result \cref{theo:lowerSDP} is stated for a special kind of polynomials, which are the block-multilinear forms.

\begin{definition}\label{def:bml}
Let $\mathcal P = \{I_1,\dots,I_t\}$ be a partition of~$[n]$ into~$t$ (pairwise disjoint) non-empty subsets.
Define the set of \emph{block-multilinear polynomials with respect to~$\mathcal P$} to be the linear subspace
\beqn
V_\mathcal P = \text{span}\big\{ x_{i_1}\cdots x_{i_t} \st i_1\in I_1,\dots, i_t\in I_t\big\}.
\eeqn
\end{definition}

We also work with the larger space of polynomials spanned by monomials where in the above we replace linearity by odd degree.

\begin{definition}\label{def:proj}
For a family $\mathcal Q\subseteq 2^{[n]}$ of pairwise disjoint subsets, let $W_\mathcal Q\subseteq \R[x_1,\dots,x_n]$ be the subspace of polynomials spanned by monomials~$x^\alpha$ with  $\alpha\in \Z_{\geq 0}^n$ satisfying 
\begin{equation} \label{eq:proj}
\sum_{i \in I} \alpha_i \equiv 1\bmod 2 \:\: \forall I\in \mathcal Q.
\end{equation}
We use $\Pi_\mathcal Q: \R[x_1,\dots,x_n]\to W_\mathcal Q$ to refer to the orthogonal projector onto $W_\mathcal Q.$
\end{definition}

\begin{remark}\label{rem:VPinvariant}
	Given a partition $\mathcal{P}$ of $[n]$, we have $V_\mathcal{P}\subset W_\mathcal{P}$. In particular,~$V_{\mathcal P}$ consists of precisely the multilinear polynomials in~$W_{\mathcal P}$.
\end{remark}

Although the projector~$\Pi_\mathcal Q$ onto~$W_\mathcal Q$ is properly defined on the space of polynomials, we will slightly abuse notation and let it act on a $t$-tensor~$T \in \R^{n\times\cdots \times n}$ as follows.
Define~$\mathcal I_\mathcal Q\subseteq [n]^t$ to be the set of $t$-tuples that contain an odd number of elements from each set~$I\in \mathcal{Q}$.
Then, we let $\Pi_\mathcal QT$ be the tensor given by
\begin{equation}\label{eq:projectionOfTensors}
	(\Pi_{\mathcal{Q}}T)_\ind{i}:=\left\{ \begin{array}{ll}
		T_{\ind{i}} & \mathrm{if}\ \ind{i}\in\mathcal{I}_\mathcal{Q},\\
		0 & \mathrm{otherwise}.
	\end{array}\right.
\end{equation}
It is not hard to see that if~$p$ is a polynomial satisfying $T(x)=p(x)$ for every $x \in \{-1,1\}^n$, then $\Pi_\mathcal{Q}T(x)=\Pi_\mathcal{Q}p(x)$ for every $x\in\{-1,1\}^n$.

We note that all the norms and seminorms we have mentioned are norms on the space $V_\mathcal{P}$ for any partition $\mathcal{P}$ of $[n]$. Hence, we can take the dual of these norms with respect to this subspace, so from now on $\norm{p}_{\infty,*}$ and $\norm{p}_{\cb,*}$ will be the dual of $\norm{p}_{\infty}$ and~$\norm{p}_{\cb}$ of $V_\mathcal{P}$, respectively. 
By contrast, when we write~$\norm{R}_{\cb,*}$ for some $t$-tensor $\mathbb{R}^{n\times\dots\times n}$ we refer to the dual norm of the completely bounded norm of~$R$ with respect to the whole space of $t$-tensors. 

We stress that $\norm{\cdot}_{\infty,*}$ need not be equal to $\norm{\cdot}_1$. This is because we are taking the dual norms with respect to $V_\mathcal{P}$ and not with respect to the space of all multilinear maps, in which case the dual norm would be $\norm{p}_1$. The following example shows that $\norm{p}_{\infty,*}\neq \norm{p}_1$ in general. 
\begin{example}\label{ex:p1neqpinfty*}
	Consider $n=3$, $t=1$, $p=(x_1+x_2+x_3)/3$ and 
    $\mathcal{P}=\{[3]\}$, so $V_\mathcal{P}$ is the set of linear polynomials. Then, $\norm{p}_1>1/3$, but $\norm{p}_{\infty,*}\leq 1/3$. Indeed, as $|p(x)|\geq 1/3$ for every $x\in\{-1,1\}^3$ and $|p(x)|>1/3$ for some $x\in\{-1,1\}^3$, we have that $\norm{p}_1>1/3$. Note that if $q$ is linear, then $\norm{\hat{q}}_1=\norm{q}_\infty$, where $\hat{q}$ is the Fourier transform of $q$. Hence
	\begin{align*}
		\norm{p}_{\infty,*}&=\sup_{q\in V_\mathcal{P},\norm{q}_\infty\leq 1}\langle p,q\rangle=\sup_{q\in V_\mathcal{P},\norm{\hat{q}}_1\leq 1}\langle \hat{p},\hat{q}\rangle\leq \sup_{\norm{\hat{q}}_1\leq 1}\norm{\hat{p}}_\infty\norm{\hat{q}}_1=\frac{1}{3},
	\end{align*}
	where in second equality we used Parseval's identity (see \cite[Chapter  1]{o2014analysis} for an introduction to Fourier analysis in the Boolean hypercube).
\end{example}


\section{$\mathcal{E}(p,t)$ for block-multilinear forms}\label{sec:maintechresult}
In this section we formally state and prove our main result:

\begin{theorem}\label{theo:lowerSDP}
	Let $\mathcal{P}$ be a partition of $[n]$ in $2t$ subsets and $p\in V_\mathcal{P}$. Then, 
	\begin{align*}
		\err(p,t)=\sup\left\{\langle p,r\rangle-\norm{r}_{\cb,*} \st r\in V_\mathcal{P},\ \norm{r}_{\infty,*}\leq 1\right\}.
	\end{align*}
\end{theorem}

For the proof, we use more convenient expressions for the completely bounded norms and the fact that the projector $\Pi_Q$ is contractive under several norms.


\subsection{Reformulation of the completely bounded norm of a form}\label{sec:reformulationcbnorm}
In this section we derive alternative expressions for the completely bounded norms that facilitate our proof of \cref{theo:lowerSDP}.
These expressions may also be useful in future applications of the of characterization quantum query algorithms in terms of completely bounded forms. After this section, we will implicitly use the formulas of \cref{prop:newexpressioncbnorm,prop:polarizationTensors} as definitions for $\norm{p}_{\cb}$ and $\norm{T}_{\cb}$, respectively.

First, we provide a simpler expression (compared to \cref{eq:cboriginal}) for $\norm{p}_{\cb}$. 

\begin{restatable}{proposition}{cbreformulation} \label{prop:newexpressioncbnorm}
	Let $p\in\mathbb{R}[x_1,\dots,x_n]_{=t}$ be a form of degree $t$. Then, 
	$$\norm{p}_{\cb}=\inf\Big\{\norm{T}_{\cb} \st p(x)=T(x)\ \forall x\in\mathbb{R}^n\Big\}.$$ 
\end{restatable}
\begin{proof}
We begin by noting that for any $t$-tensor $T$ and permutation $\sigma \in \mathcal S_t$, we have 
	\begin{align}
		(T \circ \sigma) (x) &= \sum_{\ind{i}\in [n]^t} T_{\sigma(\ind{i})}x(\ind{i}) =\sum_{\ind{i}\in [n]^t} T_{\sigma(\ind{i})}x(\sigma(\ind{i})) = T(x).\label{eq:Tosigma}
	\end{align}
    We first show that $\|p\|_{\cb} \leq  \inf \{\|T\|_{\cb} \st p(x) = T(x) \}$. By restricting in \eqref{eq:cboriginal} to decompositions where $T^{\sigma} = T/t!$ for all $\sigma$, we have 
	\begin{align*}
		\|p\|_{\cb} &= \inf \Big\{\sum_{\sigma \in \mathcal S_t} \|T^{\sigma}\|_{\cb} \st T_p = \sum_{\sigma} T^{\sigma} \circ \sigma \Big\} \\
		&\leq \inf \Big\{\|T\|_{\cb} \st T_p = \frac{1}{t!} \sum_{\sigma} T \circ \sigma \Big\}.
	\end{align*}
    Hence, it suffices to show that 
    $T_p = \frac{1}{t!} \sum_{\sigma} T \circ \sigma$ if and only if $p(x) = T(x)$ for all $x \in \R^n$.
	For the `only if' implication, assume $T_p = \frac{1}{t!} \sum_\sigma T \circ \sigma$, then for any $x \in \R^n$  
 we have that $$p(x) = T_p(x) = \frac{1}{t!} \sum_{\sigma} T \circ \sigma(x) = \frac{1}{t!} \sum_{\sigma} T(x) = T(x),$$ where in the second equality we have used \eqref{eq:Tosigma}. For the reverse implication, 
 assume that $p(x) = T(x)$ for all $x \in \R^n$. Then, using \eqref{eq:Tosigma}, we obtain \[
 \frac{1}{t!} \sum_\sigma T \circ \sigma(x) = p(x)
 \]
 for all $x \in \R^n$ and therefore $T_p = \frac{1}{t!} \sum_\sigma T \circ \sigma$ since $T_p$ is the unique symmetric $t$-tensor with $T_p(x) = p(x)$ for all $x \in \R^n$. This concludes the proof of the first inequality.

	Now for the other inequality, $\|p\|_{\cb} \geq \inf \{\|T\|_{\cb} \st p(x) = T(x) \}$, let $T^\sigma$ ($\sigma \in \mathcal S_t$) be such that $T_p = \sum_\sigma T^{\sigma} \circ \sigma$ and define $T = \sum_{\sigma} T^{\sigma}$. Then we have 
	\[
	p(x) = T_p(x) = \sum_{\sigma} T^{\sigma} \circ \sigma (x) = \sum_{\sigma} T^{\sigma}(x) = T(x),
	\]
	and $\|T\|_{\cb} \leq \sum_{\sigma} \|T^{\sigma}\|_{\cb}$ by the triangle inequality.
\end{proof}
\medskip

Second, we show that the contraction-valued maps~$A_s$ in the definition of~$\norm{T}_{\cb}$ (see Equation \eqref{eq:defTcb}) can be taken to be the same. This result can be understood as the fact that the polarization constant of completely bounded multilinear maps is~$1$.\footnote{In Banach space theory a $t$-linear map $T:X\times\dots\times X\to Y$ determines a homogeneous degree-$t$ polynomial $P:X\to Y:A\to T(A,\dots,A)$. The operator norms of $T$ and $P$ are equivalent if $T$ is symmetric: $\norm{T}\leq \norm{P}\leq K(t)\norm{T},$ where $K(t)$ is the polarization constant  of degree $t$. For a survey on the topic see \cite[Section 5.1]{moslehian2022similarities}.}

\begin{proposition}\label{prop:polarizationTensors}
	Let $T\in\mathbb{R}^{n\times \cdots \times n}$ be a $t$-tensor. Then, 
 \begin{equation*}
     \norm{T}_{\cb}=\sup\big\{\|T(A)\| \st A\colon[n] \to B_{M(d)},\   
	d\in\mathbb{N}\big\}.
 \end{equation*}
\end{proposition}

\begin{proof}
	Let $\vertiii{T}$ be the expression in the right-hand side of the statement. Note that it is the same as the expression of $\norm{T}_{\cb}$, but now the contraction-valued maps $A_1,\dots,A_t$ are all equal. This shows
	that $\vertiii{T}\leq \norm{T}_{\cb}$. To prove the other inequality, let $A_1,\dots,A_t\colon [n]\to B_{M(d)}$ and $u,v\in S^{d-1}$. Now, define the contraction-valued map $A$ by $A(i):=\sum_{s\in[t]} e_s e_{s+1}^\mathsf{T}\otimes A_s(i)$ for $i \in [n]$, and define the unit vectors $u':=e_1\otimes u$ and $v':=e_{t+1}\otimes v.$ They satisfy $$\langle u,A_1(i_1)\dots A_t(i_t)v\rangle=\langle u',A(\ind{i})v'\rangle \quad \text{for all } \ind i \in [n]^t,$$
	so in particular 
	\[
	\sum_{\ind{i}\in [n]^t}T_{\ind{i}}\langle u,A_1(i_1)\dots A_t(i_t)v\rangle =\sum_{\ind{i}\in [n]^t}T_{\ind{i}} \langle u',A(\ind{i})v'\rangle.
	\]
	Taking the supremum over all maps $A_s$ and $u,v$ shows that $\norm{T}_{\cb} \leq \vertiii{T}$,
	which concludes the proof.
\end{proof}
\medskip

\subsection{Contractivity of the projector $\Pi_Q$.}
A key element of the proof of \cref{theo:lowerSDP} is that can restrict the infimum in \cref{theo:generalexpressionforE} to the space of polynomials $W_{\mathcal Q}$ given in Definition \ref{def:proj}. 
To do that, we prove that the orthogonal projector onto this space,~$\Pi_{\mathcal Q}$ is contractive in several norms. This will follow from the fact that $\Pi_{\mathcal Q}$ has a particularly nice structure in the form of an averaging operator. Let $\mathcal{Q}$ be a family of disjoint subsets of $[n]$.
For each $I\in\mathcal{Q}$ let $z_I$ be a random variable that takes the values $-1$ and $1$ with probability $1/2$ and let $z=(z_I)_{I\in\mathcal{Q}}$.
For a bit string $x\in\{-1,1\}^n$, we define the random variable $x\cdot z \in \{-1,1\}^n$ as 
\begin{equation*}
(x\cdot z)(i):=\left\{\begin{array}{ll}
	x_iz_I& \mathrm{if\ }i\in I\ \mathrm{for\ some\ }I\in\mathcal{Q}, \\
	x_i& \mathrm{otherwise}.
\end{array}\right.    
\end{equation*}
For a matrix-valued map $A:[n]\to M(d)$ we define the random variable $A\cdot z$ in an analogous way.
\begin{proposition}\label{prop:projections}
  For any $p\in\mathbb{R}[x_1,\dots,x_n]$ and $x\in \R^n$, we have that
		\begin{align*}
			\Pi_\mathcal{Q}p(x)&=\mathbb{E}_{z} \Big[p(x\cdot z)\prod_{I\in \mathcal{Q}}z_I\Big].
		\end{align*}
		Similarly, for any $t$-tensor $T\in\mathbb{R}^{n\times \dots\times n}$, positive integer~$d$ and matrix-valued map $A\colon [n] \to M(d)$, 
        we have that 
    \begin{align*}
		\Pi_\mathcal{Q}T(A)&=\mathbb{E}_{z} \left[T(A\cdot z)\prod_{I\in \mathcal{Q}}z_I\right].
	\end{align*} 
\end{proposition}

\begin{proof}
	By linearity, it suffices to prove the equality for monomials. 
	Let $\alpha\in\mathbb{Z}_{\geq 0}^n$. Then we have 
	\[
	(x \cdot z)^\alpha  \prod_{I\in \mathcal{Q}}z_I = x^{\alpha} \prod_{I \in \mathcal Q} z_I^{1+\sum_{i \in I} \alpha_i}.
	\]
	It follows that
	\[
	\mathbb{E}_{z}\Big[(x\cdot z)^\alpha \prod_{I\in \mathcal{Q}}z_I\Big] = 
	\begin{cases} x^{\alpha} & \text{if } 1+\sum_{i \in I} \alpha_i = 0 \bmod 2 \ \forall I \in \mathcal Q, \\
	0 & \text{otherwise}.
	\end{cases}
	\]
	It remains to observe that this is precisely the projection of $x^\alpha$ on $W_{\mathcal Q}$. The statement for tensors follows analogously. 
\end{proof}

Finally, we prove that $\Pi_{\mathcal Q}$ is contractive with respect to the relevant norms.

\begin{lemma}\label{cor:projections}
	Let $\mathcal{Q}$ be a family of disjoint subsets of $[n]$ and $p\in\mathbb{R}[x_1,\dots,x_n]$ and let $\mathrm{norm}\in\{\cb,\infty,1\}$ where for the $\cb$-norm we moreover require $p$ to be homogeneous. Then $$\norm{\Pi_\mathcal{Q}p}_{\mathrm{norm}}\leq \norm{p}_{\mathrm{norm}}.$$
\end{lemma}
\begin{proof}  First, we consider the $\norm{\cdot}_\infty$ norm. For every $x\in\{-1,1\}^n$, we have that $x\cdot z\in\{-1,1\}^n$, so $$|\Pi_\mathcal{Q}p(x)|\leq \mathbb{E}_z|p(x\cdot z)\prod_{I\in\mathcal{Q}} z_I|=\mathbb{E}_z|p(x\cdot z)|\leq \mathbb{E}_z \norm{p}_{\infty}=\norm{p}_{\infty},$$
where in the first inequality we used \cref{prop:projections} and the triangle inequality. 

Second, we consider $\norm{\cdot}_\cb$. Arguing as in the $\norm{\cdot}_\infty$ case and using \cref{prop:polarizationTensors}, it follows that for any $t$-tensor $T\in\mathbb{R}^{n\times\dots\times n}$ we have that $\norm{\Pi_\mathcal{Q}T}_{\cb}\leq \norm{T}_{\cb}$. Given that $\Pi_\mathcal{Q}p(x)=\Pi_\mathcal{Q}T(x)$ if $p(x)=T(x)$, it follows that $$\norm{\Pi _\mathcal{Q}p}_{\cb}\leq \norm{\Pi_\mathcal{Q}T}_{\cb}\leq \norm{T}_{\cb} $$ for every $t$-tensor $T\in \mathbb{R}^{n\times\dots\times n}$ such that $T(x)=p(x)$. Taking the infimum over all those $T$ we arrive at $\norm{\Pi_\mathcal{Q}p}_{\cb}\leq \norm{p}_{\cb}$. 

Finally, for $\norm{\cdot}_1$ we have $$\norm{\Pi_{\mathcal Q} p}_1=\mathbb{E}_x|\mathbb{E}_{z} p(x\cdot z)\prod_{I\in \mathcal{Q}}z_I|\leq \mathbb{E}_{x}\mathbb{E}_{z} |p(x\cdot z)|=\mathbb{E}_z\mathbb{E}_x|p(x)|=\norm{p}_1,$$ where in the first equality we have used \cref{prop:projections} and in the third we have used the fact that the uniform measure is invariant under multiplication by $z\in\{-1,1\}^n$. 
\end{proof}

\subsection{Putting everything together}

We are now ready to prove \cref{theo:lowerSDP}.
	To this end, we start from the expression given in \cref{theo:generalexpressionforE} for $\err(p,t)$ and let $h \in \R[x_1,\ldots,x_{n+1}]_{=2t} \text{ with } 
	\norm{h}_{\cb}\leq 1$ and let $q:\{-1,1\}^n\to\mathbb{R}$ be defined by $q(x)=h(x,1)$ for every $x\in\{-1,1\}^n$. 
	
	We first show that we can project $q$ (and $h$) onto $W_{\mathcal P}$ and obtain a feasible solution whose objective value is at least as good as $q$. Since $\mathcal P$ is a partition of $[n]$, it defines a family of disjoint subsets of $[n+1]$, so by \cref{cor:projections}, we have $\norm{\Pi_\mathcal{P}h}_{\cb}\leq\norm{h}_{\cb}\leq 1$. 
	Since the degree of $h$ is at most $2t$, the polynomial $\Pi_{\mathcal{P}}h$ has degree at most $2t$. This shows that each monomial in its support contains exactly one 	variable from each of the $2t$ sets in $\mathcal{P}$. We can therefore observe that $\Pi_\mathcal{P}h$ does not depend on $x_{n+1}$. Since $h(x,1) = q(x)$, we have $\Pi_\mathcal{P}h(x,1)=\Pi_\mathcal{P}q(x)$ and therefore $\Pi_\mathcal{P}q\in V_\mathcal{P}$. We then use \cref{prop:newexpressioncbnorm} to show that $\norm{\Pi_\mathcal{P} q}_{\cb} \leq 1$. Indeed, applying $\Pi_\mathcal{P}$ to a $2t$-tensor $T \in \R^{(n+1) \times \cdots (n+1)}$ that certifies $\norm{h}_{\cb} \leq 1$ results in a tensor $\Pi_\mathcal{P}T$ that satisfies $\Pi_\mathcal{P}T (\ind{i}) = 0$ whenever $\ind i$ contains an index equal to $n+1$. So,  $\Pi_\mathcal{P}T(x,1)=\Pi_\mathcal{P}q(x)$ for every $x\in\{-1,1\}^n$ and thus $\Pi_{\mathcal P}T$, viewed as a $2t$-tensor in $\R^{n \times \cdots \times n}$, 
	certifies $\norm{\Pi_\mathcal{P} q}_{\cb}\leq 1$.
	For the objective value of $\Pi_{\mathcal P} q$ we finally observe that 
	$$\norm{p-\Pi_\mathcal{P}q}_\infty=\norm{\Pi_\mathcal{P}(p-q)}_\infty\leq \norm{p-q}_\infty,$$
	where  we used that $p\in V_\mathcal{P}$ in the equality and  \cref{cor:projections}  in the inequality. This shows that 
	\begin{align*}
		\mathcal{E}(p,t)\geq \inf\{ \norm{p-q}_{\infty} \st q\in V_\mathcal{P} \text{ with } \norm{q}_{\cb}\leq 1\}.
	\end{align*}
	To show that the above inequality is in fact an equality it suffices to observe that given a polynomial $q\in V_\mathcal{P}$, we can define $h \in \R[x_1,\ldots,x_{n+1}]$ 
	as $h(x,x_{n+1})=q(x)$ and then we have $\norm{h}_{\cb}\leq\norm{q}_{\cb}$. 
	
	Finally, in the above reformulation of $\err(p,t)$, we can express $\|p-q\|_\infty$ in terms of its dual norm and obtain 
	\begin{align*}
		\mathcal{E}(p,t)= \inf_q\sup_r \quad &  \langle p-q,r\rangle\\
		 \mathrm{s.t.} \quad & q\in V_\mathcal{P}\ \mathrm{with\ }\norm{q}_{\cb}\leq 1,\\
		& r\in V_\mathcal{P}\ \mathrm{with\ }\norm{r}_{\infty,*}\leq 1.
	\end{align*} 

    Finally, we need the von Neumann's minimax theorem (see \cite{Nikaid1954OnVN} for a  proof). 

    \begin{theorem}[Minimax]\label{theo:minimax}
        Let $X$ and $Y$ convex compact sets. Let $f:X\times Y\to \mathbb R$ such that $f$ is concave in the first variable and convex in the second. Then, 
        $$\sup_{x\in X}\inf_{y\in Y} f(x,y)= \inf_{y\in Y}\sup_{x\in X} f(x,y).$$
    \end{theorem}
    
	\noindent The desired result then follows by exchanging the infimum and supremum, which we are allowed to do by \cref{theo:minimax}.

\section{Separations between infinity and completely bounded norms}
\label{sec:separations}
In this section we show that the completely bounded norm of a degree 4 bounded polynomial can be unbounded. In other words, we prove the following Theorem.

\begin{theorem}\label{theo:separations}
    There is a sequence $p_n\in \mathbb R[x_1,\dots,x_n]_{=4}$ such that $$\frac{\norm{p_n}_{\cb}}{\norm{p_n}_{\infty}}\to \infty.$$
\end{theorem}

To prove \cref{theo:separations} we first provide a framework to lower bound the completely bounded norm inspired on a technique due to Varopoulos \cite{Varopoulos:1974}.\footnote{We use the same construction as the one proposed by Varopoulos, but we apply it to multilinear polynomials, which gives it the extra property displayed in \cref{eq:var1}} Second, we construct two sequences of bounded polynomials, one random and one explicit, that fit in that framework and have unbounded completely bounded norm.

\subsection{Lower bounding the completely bounded norm}\label{sec:preliminaries}

We will first talk about  general cubic forms, that is polynomials given by:
\beq\label{eq:3form}
p(x) = \sum_{S\in {[n]\choose 3}}c_S \prod_{i\in S}x_i,
\eeq
where the~$c_S$ are some real coefficients. We will lower bound its completely bounded norm. Then, we will extend this lower bound to an associated quartic form, given by $x_0p(x)$. For $i\in [n]$, define the $i$th \emph{slice} of~$p$ to be the symmetric matrix $M_i\in \R^{n\times n}$ with $(j,k)$-coefficient equal to~$c_{\{i,j,k\}}$ if $i,j,k$ are pairwise distinct and~$0$ otherwise. Then, define 
\beqn
\Delta(p) = \max_{i\in[n]} \|M_i\|.
\eeqn

\begin{lemma}[tri-linear Varopoulos decomposition]\label{lem:varopoulos}
Let~$p$ be an $n$-variate multilinear cubic form as in~\eqref{eq:3form}.
Then, for some $d\in \N$, there exists ~$A:[n]\to B_{M(d)}$ and orthogonal vectors~$u,v\in S^{d-1}$ such that $[A(j),A(i)]=0$, and
\begin{align}
A(i)^2 &=0\label{eq:var1}\\
\langle u, A(i)v\rangle &=0\label{eq:var2}\\
\langle u, A(i)A(j)v\rangle &=0\label{eq:var3}\\
\langle u, A(i)A(j)A(k)v\rangle &= \frac{c_{\{i,j,k\}}}{\Delta(p)}\label{eq:var4}
\end{align}
for all pairwise distinct $i,j,k\in [n]$.
\end{lemma}

\begin{proof}
For each $i\in[n]$, define~$M_i$ as above.
Define $W_i = \Delta(p)^{-1}M_i$ and note that $W_i$ has operator norm at most~1.
For each $i\in[n]$, define the $(2n+2)\times(2n+2)$ block matrix
\beq\label{eq:2}
A(i)
=
{\footnotesize
\left[
\begin{array}{c|c|c|c}
\phantom{M}&&&\\
\hline
e_i&\phantom{M}&&\\
\hline
&W_i^\mathsf{T}&\phantom{M}&\\
\hline
&&e_i^\mathsf{T}&\phantom{M}
\end{array}
\right]},
\eeq
where the first and last rows and columns have size~1, the second and third have size~$n$ and
where the empty blocks are filled with zeros.
Define $u = e_{2n+2}$ and $v = e_1$. \cref{eq:var2} is true because the bottom left corner of $A(i)$ is 0. From \cref{eq:2} follows that  
\beqn
A(i)A(j)
=
{\footnotesize
\left[
\begin{array}{c|c|c|c}
\phantom{x}&&&\\
\hline
&\phantom{x}&&\\
\hline
W_i^\mathsf{T}e_j&&\phantom{x}&\\
\hline
&e_i^\mathsf{T}W_j^\mathsf{T}&&\phantom{x}
\end{array}
\right]}, \ A(i)A(j)A(k)
=
{\footnotesize
\left[
\begin{array}{c|c|c|c}
\phantom{x}&&&\\
\hline
&\phantom{x}&&\\
\hline
&&\phantom{x}&\\
\hline
e_iW_j^\mathsf{T}e_k & &&\phantom{x}
\end{array}
\right]},
\eeqn
so \cref{eq:var3,eq:var4} follow by looking at the bottom left corner of these matrices. Finally, the property that~$A(i)^2 = 0$ follows from the fact that
\beqn
A(i)^2
=
{\footnotesize
\left[
\begin{array}{c|c|c|c}
\phantom{x}&&&\\
\hline
&\phantom{x}&&\\
\hline
W_i^\mathsf{T}e_i&&\phantom{x}&\\
\hline
&e_i^\mathsf{T}W_i^\mathsf{T}&&\phantom{x}
\end{array}
\right]}
\eeqn
and that the $i$th row and $i$th column of~$M_i$ (and hence $W_i$) are zero.
\end{proof}

\begin{corollary}\label{cor:ccblb}
Let~$p$ be an $n$-variate multilinear cubic form as in~\eqref{eq:3form}.
Suppose that an $(n+2)$-variate quartic form $h\in \R[x_0,x_1,\dots,x_n,z]$ satisfies $h(x,1) = x_0p(x_1,\dots,x_n)$ for every $x\in \pmset{n+1}$.
Then,
\beqn
\|h\|_{\cb} \geq \frac{\|p\|_2^2}{\Delta(p)}.
\eeqn
\end{corollary}

\begin{proof}
From the orthonormality of the characters, it follows that~$h$ and~$x_0p$ have equal coefficients for each quartic multilinear monomial in the variables $x_0,\dots,x_n$, which are~$c_S$ for~$x_0\chi_S$ with $S\in {[n]\choose 3}$ and~0 otherwise.
Let~$A:[n]\to B_{M(d)}$ and $u,v\in S^d$ be as in Lemma~\ref{lem:varopoulos}, and extend $A$ by $A(0) = \id, A(n+1) = 0$.
Commutativity and properties~\eqref{eq:var1}--\eqref{eq:var3} imply that if a quartic monomial expression $A((i,j,k,l))$ with $i,j,k,l\in \{0,\dots,n+1\}$ has repeated indices or an index equal to~$n+1$, then $\langle u, A((i,j,k,l))v\rangle = 0$.
With this, it  follows that, for every $T_h$ such that $T_h(x)=h(x)$, we have
\begin{align}\label{eq:hfeq}
\norm{T_h}_{\cb}&\geq \sum_{\ind i \in (\{0\}\cup [n+1])^4}T_{\ind i}\Big\langle u, A(\ind i)v\Big\rangle=\sum_{S\in {[n]\choose 3}} c_S\Big\langle u, \, A(0)\prod_{i\in S}A(i)v\Big\rangle.
\end{align}
Finally, if we use that $A(0)=I$, property~\eqref{eq:var4} and Parseval's identity, we obtain the desired result:
\begin{align}
\norm{h}_{\cb}=\inf\norm{T_h}_{\cb}&\geq 
 \sum_{S\in {[n]\choose 3}} c_S \langle u, \prod_{i\in S}A(i)v\rangle\nonumber=
\Delta(p)^{-1}\sum_{S\in {[n]\choose 3}}c_S^2\nonumber=
\frac{\|p\|_2^2}{\Delta(p)}.\nonumber
\end{align}
\end{proof}


\subsection{A separation based on a random example}\label{subsec:randomEx}

We begin by defining a random cubic form as in~\eqref{eq:3form} where the coefficients~$c_S$ are chosen to be independent uniformly distributed random signs.
Parseval's identity then gives $\|p\|_2^2 = {n\choose 3}$.
We now use a standard random-matrix inequality to upper bound $\Delta (p)$ (see \cite[Corollary~2.3.6]{tao2012topics} for a proof). 

\begin{lemma}\label{lem:taorandommatrix}  There exist absolute constants $C,c\in (0,\infty)$ such that the following holds. 
	Let $n$ be a positive integer and let~$M$ be a random $n\times n$ symmetric random matrix such that for $j\geq i$, the entries~$M_{ij}$ are independent random variables with mean zero and absolute value at most~1. 
	Then, for any $\tau \geq C$, we have
	\begin{align*}
	\Pr\big[\|M\|>\tau\sqrt{n}\big]\leq Ce^{-c\tau n}.
	\end{align*}
\end{lemma}

Applying \cref{lem:taorandommatrix} to the slices $M_i$ and the union bound then imply that $\Delta(p) \leq C\sqrt{n}$ with probability $1 - \exp(-Cn)$.
By Hoeffding's inequality~\cite[Theorem~2.8]{boucheron:concentration} and the union bound, we have that $\|p\|_\infty \leq Cn^2$ with probability $1 - \exp(-Cn)$.
Rescaling~$p$ then gives that there exists a bounded multilinear cubic form such that $\|p\|_2^2/\Delta(p) \geq C\sqrt{n}$. Now \cref{theo:separations} follows from \cref{cor:ccblb}.

\subsection{A construction based on an explicit example}\label{sec:explicitexample}

We also give an explicit construction using techniques from~\cite{briet2018failure}, which were used there to disprove a conjecture on a tri-linear version of Grothendieck's theorem.
We do not exactly use the construction from that paper because it involves complex functions.
Instead, we will use the M\"{o}bius function (defined below), which is real valued and has the desired properties. 

The construction uses some notions from additive combinatorics.
For a function $f:\Z_n\to[-1,1]$ (on the cyclic group of order~$n$), define the 3-linear form
\begin{equation*}
    p(x_1,x_2,x_3) = \sum_{a,b\in \Z_n}x_{1,a}x_{2, a+b}x_{3,a+2b} f(a+3b).
\end{equation*}
where $x_1,x_2,x_3\in \{-1,1\}^n$ and the sums of $a$ and $b$ are done in $\mathbb Z_n$. 

We begin by upper bounding $\Delta(p)$. The polynomial~$p$ has $3n$ slices, $M_{i,a}\in \R^{[3]\times\Z_n}$ for each $i\in[3]$ and $a\in \Z_n$, which we view as $3\times 3$ block-matrices with blocks indexed by~$\Z_n$.
The slice $M_{1,a}$ is supported only on the $(2,3)$ and $(3,2)$ blocks, which are each others' transposes.
On its $(2,3)$ block it has value $f(a+3b)$ on coordinate $(a+b, a+2b)$ for each~$b$.
In particular, this matrix has at most one nonzero entry in each row and column.
It follows that a relabeling of the rows turns $M_{1,a}$ into a diagonal matrix with diagonal entries in~$[-1,1]$, and therefore~$\|M_{1,a}\| \leq 1$.
Similarly, we get that $\|M_{i,a}\|\leq 1$ for $i=2,3$.
Hence, 
\begin{equation}\label{eq:upperboundDelta}
    \Delta(p) \leq 1.
\end{equation}
for any choice of $f$. 

Now we will choose a specific $f$ for which we will be able to upper bound $\norm{p}_{\infty}$ and lower bound $\norm{p}_2^2$. Identify~$\Z_n$ with $\{0,1\dots,n-1\}$ in the standard way.
We choose~$f$ to be the M\"{o}bius function restricted to this interval.
That is, set $f(0) = 0$ and for $a> 0$, set
\beqn
f(a) = 
\begin{cases}
1 & \text{if~$a$ is square-free with an even number of prime factors}\\
-1 & \text{if~$a$ is square-free with an odd number of prime factors}\\
0 & \text{otherwise.}
\end{cases}
\eeqn

The infinity norm of $p$ can be upper bounded in terms of the Gowers $3$-uniformity norm of $f$. This norm plays a central role in additive combinatorics and is defined by 
\begin{equation*}
    \|f\|_{U^3} = \Big(\Exp_{a,b_1,b_2,b_3\in \Z_n}\prod_{c\in \bset{3}}f(a + c_1b_1 + c_2b_2 + c_3b_3)\Big)^{\frac{1}{8}}.
\end{equation*}
The proof of the announced bound can be found in~\cite[Proposition~1.11]{Green:2007}.

\begin{lemma}[generalized von Neumann inequality]\label{lem:gvn} 
Suppose that~$n$ is coprime to~6. 
Then, for any $f:\Z_n\to [-1,1]$, we have that
\begin{equation*}
    \|p\|_\infty \leq n^2\|f\|_{U^3}.
\end{equation*}
\end{lemma}

\noindent A recent result by Tao and Ter\"{a}v\"{a}inen~\cite{tao2021quantitativej} given an upper bound to the Gowers 3-uniformity norm of the M\"obius function. 
\begin{theorem}\label{theo:TaoTer}
    Let $f:\mathbb Z_n\to \mathbb R$ be the M\"obius function. Then, 
    \beqn
\norm{f}_{U^3}\leq \frac{1}{(\log \log n)^C}.
\eeqn
for some constant~$C>0$.
\end{theorem}
\noindent Combining \cref{lem:gvn,theo:TaoTer} it follows that 
\begin{equation}\label{eq:pinftyupperbound}
    \norm{p}_{\infty}\leq\frac{n^2}{(\log \log n)^C} 
\end{equation}
for some constant $C>0.$

To lower bound $\norm{p}_2^2$ we begin using Parseval's identity, which implies that
\beq\label{eq:ParsevalMobius}
\|p\|_2^2 = n\sum_{a\in \Z_n}f(a)^2.
\eeq
Given that $|f(a)|^2$ is 1 if $a$ is square-free and $0$ otherwise, we can use a classical result of number theory to lower bound $\norm{p}_2^2$ (see \cite[page 269]{hardy1979introduction} for a proof). 

\begin{proposition}\label{prop:square-free}
    There are $\frac{6}{\pi^2}n-O(\sqrt{n})$ natural numbers between 1 and $n$ that are square-free. 
\end{proposition}

\noindent From \cref{eq:ParsevalMobius,prop:square-free} follows that

\beq\label{eq:p2lowerbound}
\|p\|_2^2 = \frac{6}{\pi^2}n^2-O(\sqrt{n^3}).
\eeq

Finally,  we substitute $p$ by $p/(n^2/(\log \log n)^C)$, and it follows from \cref{eq:pinftyupperbound,eq:p2lowerbound,eq:upperboundDelta} that $p$ is bounded and 

\beqn
\frac{\|p\|_2^2}{\Delta(p)} \geq \frac{6}{\pi^2}(\log\log n)^C - o(1).
\eeqn
Again, \cref{theo:separations} now follows from \cref{cor:ccblb}.

\begin{remark}
The \emph{jointly completely bounded norm} of~$p$ is given by
\beqn
\|p\|_{\jcb} = \sup_{d\in \N} \|\sum_{a,b\in \Z_n}A(1,a)A(2, a+b)A(3,a+2b) f(a+3b)\|, 
\eeqn
where the supremum is taken over maps $A:[3]\times [n]\to \C^{d\times d}$ such that  $\|A(i,a)\|\leq 1$ and $[A(i,a),A(j,b)] = [A(i,a),A(j,b)^*] = 0$ for all ${i\ne j}$ and $a,b\in \Z_n$.
This norm can also be stated in terms of tensor products and the supremum is attained by observable-valued maps. 
As such, this norm appears naturally in the context of non-local games.
It was shown in~\cite{BBBLL:2019} that Proposition~\ref{lem:gvn} also holds for the jointly completely bounded norm, that is $\|p\|_{\jcb} \leq n^2\|f\|_{U^3}$.
The proof of Corollary~\ref{cor:ccblb} easily implies that $\|p\|_{\cb} \geq \|p\|_2^2/\Delta(p)$.
This was used in~\cite{briet2018failure} to prove that the $\jcb$ and $\cb$ norms are inequivalent.
\end{remark}
\section{Grothendieck inequalities characterize converses to the polynomial method}\label{sec:noconverses}

In this section we show, as a corollary of our main result \cref{theo:lowerSDP}, that Grothendieck inequalities characterize converses to the polynomial method . By this we mean that: i) for 1-query algorithms an additive converse is possible 
and moreover this converse characterizes $K_G$; and ii) for 2-query algorithms no additive converse is possible, because Grothendieck's inequality fails for 3-linear forms.

\subsection{Characterizing $K_G$ with $1$-query quantum algorithms}\label{sec:characterizingKG}

Here we prove \cref{theo:characterizingKG}. Before doing that, we should prove \cref{cor:cbbilinearforms}, which states that the completely bounded norm of a bilinear form $p$ regarded as polynomial is equal to its completely bounded norm as matrix $A$. This is not obvious from \cref{eq:cboriginal},
because for a~$t$-tensor~$T$ and permutation~$\sigma \in \mathcal S_t$, it is in general not true that $\|T\|_{\cb}=\|T\circ \sigma\|_{\cb}$ (see for instance~\cite{QQA=CBF}). 
It is not hard to show however, that when $T$ is a matrix we have $\|T^{\mathsf T}\|_{\cb}= \|T \|_{\cb}$. This gives the following reformulation of the completely bounded norm of forms of degree $2$. 

\begin{proposition} \label{prop: cb quadratic}
	Let $p \in \R[x_1,\ldots,x_n]_{=2}$ be a quadratic form and let $T_p$ be the unique symmetric matrix associated to $p$ via \eqref{eq:fromCjtoTp}. Then 
	$\|p\|_{\cb} = \|T_p\|_{\cb}$.
\end{proposition}

\begin{proof}
	Let $T\in\mathbb{R}^{n\times n}$ be a matrix. 
	First, we have 
	\begin{align}
		\label{eq: cb of transpose}
		\|T^{\mathsf T} \|_{\cb} &= \sup \Big\{\|\sum_{i,j} T_{j,i} A(i) B(j) \| \st A,B\colon [n]\to B_{M(d)}\Big\}  \\
		&= \sup \Big\{\|\sum_{i,j} T_{j,i}  B(j)^{\mathsf T} A(i)^{\mathsf T} \| \st A,B\colon [n]\to B_{M(d)}\Big\} \notag \\
		&= \|T\|_{\cb}, \notag
	\end{align}
	where we use (twice) that for any matrix $M$ we have $\|M\| = \|M^{\mathsf T}\|$. 

	Let $T\in\mathbb{R}^{n\times n}$ be a matrix with $p(x)=T(x)$. Then, $T_p=(T+T^\mathsf{T})/2$, so from the above and the triangle inequality it follows that 	
	\[
	\|T_p\|_{\cb} = \frac{1}{2}\norm{T + T^{\mathsf T}}_{\cb} \leq \frac{1}{2}(\|T\|_{\cb} + \|T^{\mathsf T}\|_{\cb}) = \|T\|_{\cb}.
	\]
	Using Proposition \ref{prop:newexpressioncbnorm} we conclude that $\norm{p}_{\cb}=\norm{T_p}_{\cb}$.
\end{proof}
Considering bilinear forms gives the following corollary. 
\begin{corollary}\label{cor:cbbilinearforms}
	Let $p:\{-1,1\}^n\times\{-1,1\}^n\to \R$ be a bilinear form and let $A\in\mathbb{R}^{n\times n}$ be such that $p(x,y)=x^{\mathsf T}Ay$ for all $x,y\in \R^{n}$. 
	Then, $\norm{A}_{\cb}=\norm{p}_{\cb}$.
\end{corollary}
\begin{proof}
	By \cref{prop: cb quadratic} we have $\|p\|_{\cb} = \|T_p\|_{\cb}$. Now observe that $T_p = \frac{1}{2} \begin{pmatrix} 0 & A \\ A^{\mathsf T} & 0 \end{pmatrix}$ and hence 
	\begin{align*}
		\|T_p\|_{\cb} &\leq \frac{1}{2} \left(\Big\|\begin{pmatrix} 0 & A \\ 0 & 0 \end{pmatrix}\Big\|_{\cb} + \Big\|\begin{pmatrix} 0 & 0 \\ A^{\mathsf T} & 0 \end{pmatrix}\Big\|_{\cb}\right) \\
		&\leq \frac{1}{2} \left(\|A\|_{\cb} + \|A^{\mathsf T}\|_{\cb}\right) \\
		&=\|A\|_{\cb},
	\end{align*}
	where the last equality uses \eqref{eq: cb of transpose}. 

	Conversely, let $B,C\colon[n]\to B_{M(d)}$. Note that $$\norm{\sum_{i,j\in [n]} A_{i,j}B(i)C(j)} = \sup_{u,v} \langle u, \sum_{i,j\in [n]} A_{i,j}B(i)C(j) v\rangle$$ where the supremum is taken over possibly distinct unit vectors $u$ and $v$. We argue that we can take $u=v$ without loss of generality. Indeed, let $U$ be a unitary matrix for which $U u = v$, then 
 \[
 \norm{\sum_{i,j\in [n]} A_{i,j}B(i)C(j) U} = \norm{\sum_{i,j\in [n]} A_{i,j}B(i)C(j)} = \langle u, \sum_{i,j\in [n]} A_{i,j}B(i)C(j) Uu\rangle.
 \]
 As we are optimizing over all $B,C\colon[n]\to B_{M(d)}$, replacing each $C(j)$ by $C(j) U$, we may assume that

	$$\norm{\sum_{i,j\in [n]} A_{i,j}B(i)C(j)} =\Big|\sum_{i,j\in [n]} A_{i,j}\langle u,B (i)C (j)u\rangle\Big|,$$ 
	for a single unit vector $u$. 
	 For every $i \in [n]$, define the following contractions:
	\[
	Q_1(i) = B (i) = Q_2^{\mathsf T}(i), \qquad Q_2(n+i) = C (i) = Q_{1}^{\mathsf T}(n+i).
	\]
	Then,
	\begin{align*}
		\|T_p\|_{\cb} &\geq \Big\|\sum_{i,j\in[n]}  \frac{1}{2}\big(A_{i,j} Q_1(i) Q_2(n+j) + A_{j,i} Q_1(n+i) Q_2(j)\big) \Big\| \\ 
		&= \Big|\sum_{i,j\in[n]} A_{i,j} \langle u,\frac{B(i)C(j) + C^{\mathsf T}(j)B^{\mathsf T}(i)}{2} u\rangle \Big| \\
		&=\Big|\sum_{i,j\in [n]} A_{i,j} \langle u,B(i)C (j)u\rangle \Big|
	\end{align*}
	Taking the supremum over $B,C\colon[n]\to B_{M(d)}$ shows that $\|T_p\|_{\cb} \geq \|A\|_{\cb}$.
\end{proof}

We recall that it was shown in \cite{Aaronson2015PolynomialsQQ} that for every bilinear form there exists a 1-query quantum algorithm that makes additive error at most $1-1/K_G$. It thus remains to show the lower bound. 
\characterizingKG*
\begin{proof} 
\cref{theo:lowerSDP} shows the following:
\begin{align}\label{eq:almostdone}
		\sup_{p\in\mathcal{BB}}\mathcal{E}(p,1)&=\sup_{\norm{p}_\infty\leq 1}\sup_{\norm{r}_{\infty,*}\leq 1}\langle p,r\rangle-\norm{r}_{\cb,*}\\
		&=\sup_{\norm{r}_{\infty,*}\leq 1}\norm{r}_{\infty,*}-\norm{r}_{\cb,*}\\
		&=\sup_{\norm{r}_{\infty,*}= 1}1-\norm{r}_{\cb,*}.\nonumber
\end{align}
It thus remains to show that for bilinear forms  $\|r\|_{\infty,*} \leq K_G \|r\|_{\cb,*}$. We do so starting from Grothendieck's theorem for matrices. It states that for $A \in \R^{n \times n}$ we have $\|A\|_{\cb} \leq K_G \|A\|_{\infty}$. Each bilinear form $q:\{-1,1\}^n \times \{-1,1\}^n \to \R$ uniquely corresponds to a matrix $A \in \R^{n \times n}$ such that $q(x,y) = x^{\mathsf T}Ay$. Moreover, for such $q$ and $A$ one has $\|q\|_\infty = \|A\|_{\infty}$ (immediate) and in \cref{cor:cbbilinearforms} we showed $\|q\|_{\cb} = \|A\|_{\cb}$, so $\norm{q}_\cb\leq K_G \norm{q}_\infty$. A duality argument then concludes the proof:
\[
\|r\|_{\infty,*} = \sup_{\|q\|_{\infty} \leq 1} \langle r,q \rangle \leq \sup_{\|q\|_{\cb} \leq K_G} \langle r,q\rangle = K_G \|r\|_{\cb,*}. 
\]
\end{proof}
\begin{remark}
    If in \cref{theo:characterizingKG} we restrict the supremum to bilinear forms on $n+n$ variables, for a fixed $n$, then we obtain a characterization of $K_G(n)$ instead of $K_G$. Here, $K_G(n)=\sup\norm{A}_{\cb}/\norm{A}_{\infty}$, where the supremum is taken over all non-zero $n\times n$ real matrices.
\end{remark}

\subsection{No converse for the polynomial method}\label{sec:NoAdditiveConverse}
In this section we show that there is no additive nor multiplicative converse for polynomials of degree 4 and 2-query algorithms. In other words, we will prove \cref{theo:NoAdditiveConverse,theo:NoMultConverse}.
Before doing that, we explain what was the error in the proof of \cref{theo:NoMultConverse} given in \cite{QQA=CBF}. 

Their proof arrives at the equation 
\begin{equation}\label{eq1}
    \sum_{\alpha,\beta\in\{0,1,2,3,4\}^{n}: |\alpha|+|\beta|=4}d'_{\alpha,\beta}x^\alpha=C\sum_{\alpha\in\{0,1\}^{n}: |\alpha|=4}d_\alpha x^\alpha\ \ \forall\ x\in\{-1,1\}^{n},
\end{equation}
where $d'_{\alpha,\beta}$, $d_\alpha$ and $C$ are some real numbers, $x^\alpha$ stands for $\prod_{i=1}^n x_i^{\alpha_i}$ and $|\alpha|$ for $\sum_{i=1}^n\alpha_i$. 
It follows from the orthogonality of the characters that $d'_{\alpha,0} = Cd_\alpha$ for all $\alpha\in\{0,1\}^n$ such that $|\alpha| = 4$.
What is used, however, is that $d'_{\alpha,0} = Cd_\alpha$ for all $\alpha\in\{0,1,2,3,4\}^n$ such that $|\alpha| = 4$, which is not true in general. 
For instance if $n=1$, $C=1$ and $d'_{(2,0),(0,2)}=-d'_{(0,0),(4,0)}=1$ and the rest of the coefficients set to~$0$, then~\eqref{eq1} becomes $x^2-1=0,\ \forall\ x\in\{-1,1\}$. 

We now prove that there is no additive converse, from which the non-multiplicative converse result quickly follows.

\NoAdditiveConverse*
\begin{proof} 
	For any partition $\mathcal P$ of $\{0\}\cup [3n]$ in $2t$ subsets, \cref{theo:lowerSDP} shows that 
	\begin{align*}
		\sup_{p\in V_{\mathcal{P}},\norm{p}_\infty\leq 1}\mathcal{E}(p,t) &=\sup_{p\in V_{\mathcal{P}},\norm{p}_\infty\leq 1}\sup_{r\in V_{\mathcal{P}},\norm{r}_{\infty,*}\leq 1}\langle p,r\rangle-\norm{r}_{\cb,*}\\
		&=\sup_{r\in V_{\mathcal{P}},\norm{r}_{\infty,*}\leq 1}\norm{r}_{\infty,*}-\norm{r}_{\cb,*}\\
		&=\sup_{r\in V_{\mathcal{P}},\norm{r}_{\infty,*}= 1}1-\norm{r}_{\cb,*}.\nonumber
	\end{align*}
	Consider now the case $t=2$ and the partition $\mathcal{P}_n=\{\{0\},\{1,\dots,n\},\{n+1,\dots,2n\},\{2n+1,\dots,3n\}\}$ of $\{0\}\cup[3n]$. 
	In \cref{theo:separations} a sequence of forms $p_n\in V_{\mathcal{P}_n}$ was constructed with the property that 
	 \begin{equation}\label{eq:cbinftyquotient}
	 	\frac{\norm{p_n}_{\cb}}{\norm{p_n}_\infty}\to \infty.
	 \end{equation}
	 Hence, by a duality argument we get that there is a sequence $r_n\in V_{\mathcal{P}_n}$
	 such that $\norm{r_n}_{\cb,*}/\norm{r_n}_{\infty,*}\to 0$. Indeed, suppose towards a contradiction that there is a $K>0$ such that for every $n\in\mathbb{N}$ and every $r\in V_{\mathcal{P}_n}$ we have that $\norm{r}_{\cb,*}\geq K\norm{r}_{\infty,*}$. Then, $$\norm{p}_{\cb}=\sup_{\norm{r}_{\cb,*}\leq 1}\langle r,p\rangle\leq\frac{1}{K}\sup_{\norm{r}_{\infty,*}\leq 1}\langle r,p\rangle=\frac{1}{K}\norm{p}_\infty,$$ which contradicts \cref{eq:cbinftyquotient}. The sequence $r_n$ shows that 
	  $$\sup_{p\in V_{\mathcal{P}_n},\norm{p}_\infty\leq 1,n\in\mathbb{N}}\mathcal{E}(p,2)=1,$$ which implies the stated result.
\end{proof}
\NoMultConverse*

\begin{proof}
    First note that we can assume $C\leq 1$, because  $|\mathbb E[\mathcal{A}(x)]|\leq 1$ for any algorithm $\mathcal{A}$ and any $x\in \{-1,1\}^n$. Assume that there exists $0<C\leq 1$ such that for every bounded $p$ of degree 4 there is a 2-query algorithm $\mathcal{A}$ with $\mathbb E[\mathcal{ A}(x)]=p(x)$ for every $x\in\{-1,1\}^n$. We claim that that $\mathcal{A}$ approximates $p$ up to an additive error $1-1/C$, which contradicts \cref{theo:NoAdditiveConverse}. Indeed, 
    \[
        |p(x)-\mathbb E[\mathcal{A}(x)]|=|p(x)(1-C)|\leq 1-C. 
    \]
\end{proof}

\section{An open question}\label{sec:question}
Let $\mathcal{P}$ be a partition of $[n]$ in $2t$ subsets and let $p\in V_\mathcal{P}$ with $\norm{p}_\infty\leq 1$.
From the characterization of quantum $t$-query algorithms of \cite{QQA=CBF} we know that there is a quantum query algorithm~$\mathcal A$ that outputs $p/\norm{p}_{\cb}$ on expectation. 
 In particular, 
 $$|\mathbb{E}[\mathcal{A}(x)]-p(x)|=\left|\frac{p(x)}{\norm{p}_{\cb}}-p(x)\right|\leq \norm{p}_{\infty}\Big(1-\frac{1}{\norm{p}_{\cb}}\Big).$$
As a consequence, one has that 
\begin{equation}\label{eq:almostGT}
	\mathcal{E}(p,t)\leq \norm{p}_{\infty}\Big(1-\frac{1}{\norm{p}_{\cb}}\Big).
\end{equation}
Our \cref{theo:lowerSDP} implies that both sides of Equation \eqref{eq:almostGT} are equal when you take the supremum over all $p\in V_\mathcal{P}$. We wonder if that is true for every $p\in V_\mathcal{P}$. 
\begin{question}\label{que:openquestion}
	Let $\mathcal{P}$ be a partition of $[n]$ in $2t$ subsets and let $p\in V_\mathcal{P}$. Is it true that 
	$$	\mathcal{E}(p,t)= \norm{p}_{\infty}\Big(1-\frac{1}{\norm{p}_{\cb}}\Big)?$$
\end{question}

A positive answer to this question would strengthen our main technical contribution, Theorem~\ref{theo:lowerSDP}. 
 In addition, if we focus on the case $t=1$, it would imply that the method proposed in \cite{Aaronson2015PolynomialsQQ} to give an algorithm that computes $p/\norm{p}_{\cb}$ with $1$ query, provides the best $1$ query approximation for every $p$ (here the best means the one that minimizes $\mathcal{E}(p,1)$). For $t\geq 1$ it was showed in \cite{QQA=CBF} that $p/\norm{p}_\cb$ is the output of a $t$ query algorithm, which would also be the best possible $t$ query approximation for $p$.
 
 Finally, for the case $t=1$, it would imply a clean link between the biases of two player XOR games and quantum query algorithms. Indeed, given a matrix $A\in\mathbb{R}^{n\times n}$ it both defines a bounded bilinear form $p_A(x,y)=x^{\mathsf T} A y $ and a two player XOR game $G_A$, where the referee asks the pair of questions $(i,j)$ with probability $$\pi(i,j)=\frac{|A_{i,j}|}{{\sum_{i,j\in [n]}|A_{i,j}|}}$$ and the payoff is given by $$\mu(i,j,a,b)=\frac{1+ab\cdot\mathrm{sgn}[A_{i,j}]}{2}.$$ 
\cref{cor:cbbilinearforms} states that $$\norm{p_A}_\infty=\norm{A}_\infty\ \mathrm{and}\ \norm{p_A}_{\cb}=\norm{A}_{\cb},$$
while Tsirelson's work \cite{Tsirelson} implies that the classical and quantum biases of $G_A$ are
$$\beta(G_A)=\norm{A}_\infty\ \mathrm{and}\ \beta^*(G_A)=\norm{A}_{\cb}.$$
Thus, a positive answer to \cref{que:openquestion} would imply that 
$$\mathcal{E}(p_A,1)=\beta(G_A)\Big(1-\frac{1}{\beta^*(G_A)}\Big).$$

\section{An approach via duality}\label{sec:dualnorms}


In this section we give an alternative proof of our main result from the dual picture. We stress that our main result gives a formula for $\err(p,t)$ in terms of certain dual norms, but it is not clear yet how to compute them. Along the way to the alternative proof of our main result, we will give formulations of the dual norms as efficiently solvable convex optimization problems. In particular, we will show that~$\|p\|_{\cb,*}$ can be computed using a semidefinite program and~$\|p\|_{\infty,*}$ via a linear program. The semidefinite programming formulation of~$\|p\|_{\cb,*}$ essentially follows from a semidefinite programming formulation for the completely bounded norm of a tensor that was provided in \cite{gribling2019semidefinite}. However, proving this, and using it to provide an alternative proof of \cref{theo:generalexpressionforE}, requires us to further develop some of the theory of completely bounded norms of polynomials.

\subsection{How to compute the dual norms} 
First, we give a convenient formula for the dual of the completely bounded norm of a tensor with respect to the inner product given by 
$$\langle T,R\rangle=\sum_{\ind{i}\in[n]^t}T_{\ind{i}}R_{\ind{i}}.$$ To state it in an compact way, we introduce the following subset of $t$-tensors in~$\R^{n\times\cdots\times n}$:
\begin{equation} \label{eq:Knt}
\begin{split}
	K(n,t) := \{ \langle u,A(\ind{i})v\rangle_{\ind{i}\in[n]^t} \st  d\in\mathbb{N},\ u,v\in S^{d-1},\ A\colon [n] \to B_{M(d)} \}.
\end{split}
\end{equation}

\begin{proposition}\label{prop:Tcb*}
	Let $R\in \mathbb{R}^{n\times \dots\times n}$ be a $t$-tensor. Then, \begin{align} \label{eq:dualityTcb}
		\|R\|_{\cb,*} = \inf\{w>0 \st R\in wK(n,t)\}.
	\end{align}
\end{proposition}
\begin{remark}
	Note that \cref{prop:Tcb*} states that $\norm{R}_{\cb,*}$ is the Minkowski norm defined by $K(n,t)$, so $K(n,t)$ is the unit ball of $(\mathbb{R}^{n\times \dots\times n},\norm{\cdot}_{\cb,*})$.
\end{remark}
\begin{remark}
     \cref{prop:Tcb*} allows us to compute $\norm{R}_{\cb,*}$ as a SDP. Indeed, it follows from \cref{eq:dualityTcb} that 
    \begin{align*}
        \norm{R}_{\cb,*}=\ \ &\inf && w \nonumber\\
                        &\  \mathrm{s.t.} && w\in\mathbb R_{>0},\, d\in\mathbb{N},\,  u,v\in wS^{d-1},\ A\colon [n] \to B_{M(d)},\nonumber\\  & && 
                        R = \langle u,A(\ind{i})v\rangle_{\ind{i}\in[n]^t},\  
    \end{align*}    
    and in \cite{gribling2019semidefinite} it was shown (implicitly) that the constraints can be written as linear equations on the entries of a positive semidefinite matrix. 
\end{remark}
\begin{proof}[ of \cref{prop:Tcb*}]
	Let $\vertiii{R}$ be the expression in the right-hand side of \cref{eq:dualityTcb}. First, we show that $\vertiii{\cdot}$ is a norm. We should check that $\vertiii{R}$ is well-defined, i.e., that every tensor can be decomposed as $\langle u, A(\ind{i})v\rangle$. We observe that the standard basis elements for the space of $t$-tensors are contained in $K(n,t)$. Indeed, let
	$u:=e_1$, $v:=e_{t+1}$ and $A(i)=\sum_{s\in\ind{i}^{-1}(i)} e_s e_{s+1}^\mathsf{T}$ then 
	\[
	\langle u, A(\ind{j})v\rangle = \begin{cases} 1 &\text{ if } \ind j = \ind i,\\ 0 &\text{ otherwise.} \end{cases}
	\] 
	To conclude that $\vertiii{R}$ is well-defined it then suffices to observe that the set of scalar multiples of elements in $K(n,t)$ is closed under addition. Indeed, if 
	\begin{align*}
		R_{\ind{i}}=\langle u,A(\ind{i})v\rangle\ \mathrm{and}\ \tilde{R}_{\ind{i}}=\langle \tilde{u},\tilde{A}(\ind{i})\tilde{v}\rangle,
	\end{align*} 
	for some $u,v,\tilde{u},\tilde{v}\in\mathbb{R}^d$ with $\norm{u}^2=\norm{v}^2=w$, $\norm{\tilde{u}}^2=\norm{\tilde{v}}^2=\tilde{w}$ and maps $A,\tilde{A}\colon[n]\to B_{M(d)},$ then 
	\begin{equation*}
		R_\ind{i}+\tilde{R}_{\ind{i}}=\langle \hat{u},\hat{A}(\ind{i})\hat{v}\rangle,
	\end{equation*}
	where $\hat{u},\hat{v}\in\mathbb{R}^{2d}$ are the vectors with $\norm{\hat{u}}^2=\norm{\hat{v}}^2=w+\tilde{w}$ defined by $\hat{u}:= u \oplus \tilde{u}$, $v:= v \oplus \tilde{v}$ and the map $\hat{A}\colon [n] \to B_{M(2d)}$ is defined via 
	$$\hat{A}(i)=\begin{pmatrix}
		A(i) & 0 \\ 0 &\tilde{A}(i)
	\end{pmatrix}.$$
	This construction also shows that $\vertiii{\cdot}$ satisfies the triangle inequality. It is also clear that $\vertiii{\cdot}$ is homogeneous and that $\vertiii{R}=0$ if and only if $R=0$, so $\vertiii{\cdot}$ is a norm. 

	Finally, note that the completely bounded norm of a $t$-tensor $R\in \mathbb{R}^{n\times \dots\times n}$ is given by $$\norm{T}_{\cb}=\sup\Big\{\big|\sum_{\ind{i}\in [n]^t}T_{\ind{i}}\langle u,A(\ind{i})v\rangle\big|  \st d \in \N,\  u,v \in S^{d-1},\  A\colon[n] \to B_{M(d)}\Big\},$$
	so $\norm{T}_{\cb}=\vertiii{T}_{*}$, and by the fact that the dual of the dual norm is the primal norm for finite-dimensional normed spaces, we conclude that $$\vertiii{\cdot}=\vertiii{\cdot}_{**}=\norm{\cdot}_{\cb,*}.$$
\end{proof}


\begin{proposition}\label{prop:dualityinfty}
		Let $\mathcal{P}$ be a partition of $[n]$ in $t$ subsets and $p\in V_{\mathcal P}$. Then,
		\begin{equation}\label{eq:dualityinfty}
			\|p\|_{\infty,*} = 	\inf\{\|r\|_1\st\, r:\{-1,1\}^n\to \mathbb{R},\, r\in W_{\mathcal{P}},\, r_{=t} = p\big\}.
		\end{equation}
\end{proposition}
\begin{remark}
    \cref{eq:dualityinfty} can be phrased as a linear program, so it provides an efficient way of computing $\norm{p}_{\infty,*}$.
\end{remark}
\begin{proof}
    We prove the statement in a few steps:
    \begin{align*}
        \norm{p}_{\infty,*}&=\sup\big\{\langle p,q\rangle \st q\in W_\mathcal{P},\, \|q\|_{\infty} \leq 1\} \notag\\
        &=\sup\big\{\langle p,q\rangle \st q\in \mathbb{R}[x_1,\dots,x_n]_{=t},\, \|q\|_{\infty} \leq 1\} \notag\\
        &=\inf\big\{\|r\|_1 \st r:\{-1,1\}^n\to \mathbb{R},\,  r_{=t}= p\big\}\\
        &=\inf\big\{\|r\|_1 \st r:\{-1,1\}^n\to \mathbb{R},\, r\in W_p,\,  r_{=t}= p\big\},
    \end{align*}
    where the first equality is the definition, in the second equality we have used  \cref{cor:projections} with $\norm{\cdot}_{\infty}$ to remove the condition $q\in W_{\mathcal{P}}$, the third equality follows from Lagrange duality (cf.~\cite[Sec.~5.1.6]{boyd_vandenberghe_2004}), and the fourth from \cref{cor:projections} for $\norm{\cdot}_1$.
\end{proof}

We have already seen in \cref{ex:p1neqpinfty*} that $\|p\|_{\infty,*} \neq \|p\|_1$ in general, 
because we are taking the dual norm with respect to $V_\mathcal{P}$. For completeness we give an alternative proof of the separation using \cref{prop:dualityinfty}. 
\begin{example}
The upper bound  $\norm{p}_{\infty,*}\leq 1/3$ of Example \ref{ex:p1neqpinfty*}  follows from \cref{prop:dualityinfty} by considering the multilinear map $r(x)=(x_1+x_2+x_3+x_1x_2x_3)/3$ that belongs to $W_\mathcal{P}$, satisfies $r_{=1}(x)=p(x)=(x_1+x_2+x_3)/3$ and $\norm{r}_1= 1/3$.
\end{example}

\begin{proposition}\label{prop:dualitycb}
	Let $\mathcal{P}$ be a partition of $[n]$ in $t$ subsets and let $p\in V_{\mathcal P}$. Then,
	\begin{align} \label{eq:dualitypcb}
		\|p\|_{\cb,*} =t!\norm{T_p}_{\cb,*}.
	\end{align}
\end{proposition}
\begin{proof}
	By duality and definition of $\norm{\cdot}_{\cb}$, we have that 
	\begin{align*}
		\norm{p}_{\cb,*}&=\sup\{ \sum_{\alpha\in\mathbb{Z}_{\geq 0}^n}c_\alpha c_\alpha': q\in V_{\mathcal P},\, \norm{q}_{\cb}\leq 1\}\\
		&=\sup\{ \sum_{\alpha\in\mathbb{Z}_{\geq 0}^n}c_\alpha \sum_{\ind{i}\in\mathcal{I}_\alpha}T_{\ind{i}}:\, q\in V_{\mathcal P},\, T(x)=q(x),\, \norm{T}_{\cb}\leq 1\},
	\end{align*}
	where $c_\alpha$ and $c_\alpha'$ are the coefficients of $p$ and $q$, respectively. Now, let $R\in\mathbb{R}^{n\times \dots\times n}$ be a $t$-tensor such that $R(x)=p(x)$ for every $x\in\mathbb{R}^n$. 
	Then we have, using \cref{eq:fromTensorsToCoefficients}, that 
	\begin{align*} 
		\norm{p}_{\cb,*}&=\sup\{\sum_{\alpha\in\mathbb{Z}_{\geq 0}^n} \sum_{\ind{j}\in\mathcal{I}_\alpha}R_\ind{j}\sum_{\ind{i}\in\mathcal{I}_\alpha}T_{\ind{i}}: q\in V_{\mathcal P},\, T(x)=q(x),\, \norm{T}_{\cb}\leq 1\} .
	\end{align*}
	In particular, if we choose $R$ to be $T_p$, then we have 
	\begin{align*} 
		\norm{p}_{\cb,*}&=t! \sup\{\sum_{\alpha\in\mathbb{Z}_{\geq 0}^n} \sum_{\ind{i}\in\mathcal{I}_\alpha}(T_p)_{\ind i} T_{\ind{i}}: q\in V_{\mathcal P},\, T(x)=q(x),\, \norm{T}_{\cb}\leq 1\}	\\
		&=t!\sup\{ \langle T_p,T\rangle: q\in V_{\mathcal P},\, T(x)=q(x),\, \norm{T}_{\cb}\leq 1\}.
	\end{align*} 
	 We now show that the expression on the right equals $t!$ times $\norm{T_p}_{\cb,*}$, which we recall can be written as 
	\begin{align}\label{eq:Tcb*}
		\norm{T_p}_{\cb,*}=\sup\{ \langle T_p,T\rangle: \norm{T}_{\cb}\leq 1\}.
	\end{align}
	By inclusion of the feasible region we have that $\norm{p}_{\cb,*}\leq t!\norm{T_p}_{\cb,*} $. For the other inequality, let $T\in\mathbb{R}^{n\times \dots\times n}$ be a $t$-tensor and consider $\Pi_\mathcal{P} T$ as in \cref{eq:projectionOfTensors}.
	By \cref{prop:projections} we have $\norm{\Pi_\mathcal{P}T}_{\cb}\leq \norm{T}_{\cb} \leq 1$. Also note that the polynomial $\Pi_\mathcal{P}T(x)$ belongs to $V_\mathcal{P}$, because $(\Pi_\mathcal{P}T)_\ind{i}=0$ unless $\ind{i}$ contains exactly one index from each set in the partition $\mathcal{P}$. It remains to observe that $(\Pi_\mathcal{P}T)_\ind{i}= T_{\ind{i}}$ for all indices $\ind{i} \in [n]^t$ for which $(T_p)_{\ind i} \neq 0$ and therefore 
	\[
\langle T_p,T\rangle=\langle T_p,\Pi_\mathcal{P}T\rangle.
\]
This shows that $\norm{p}_{\cb,*}\geq t!\norm{T_p}_{\cb,*}$. 
\end{proof}
\subsection{Alternative proof of the main result via semidefinite programming}\label{ap:dualproof}
First of all, we will state \cref{theo:SDP} (which corresponds, after some reformulation, to equation (20) of~\cite{gribling2019semidefinite}), that gives an optimization problem equivalent to the dual of the SDP $\err(p,t)$. Before that, we introduce the following notation. Given $\ind{i}\in [n+1]^{2t}$,  $\alpha(\ind{i})\in \{0,1\}^n$ is defined as $$(\alpha(\ind{i}))_{m}:=\left\{\begin{array}{ll}1 & \mathrm{if}\ m\in[n]\ \mathrm{and}\ m\ \mathrm{occurs\ an\ odd\ number\ of\ times\ in\ }\ind{i},\\ 0 & \mathrm{otherwise}. \end{array}\right. $$

\begin{theorem}[\cite{gribling2019semidefinite}] \label{theo:SDP}
	Let $p:\{-1,1\}^n\to \mathbb{R}$ and $t\in\mathbb{N}$. Then,
	\begin{align}\label{eq:SDP}
		\err(p,t)=\sup&\ (\langle p,r\rangle-w)/\norm{r}_1\\
		\mathrm{s.t.}&\ r:\{-1,1\}^n\to\mathbb{R},\  d\in\mathbb{R}\nonumber\\
		&\ A_s\colon [n+1]\to B_{M(d)}\ \mathrm{for\ all}\ s\in[2t]\nonumber\\
		&\ u,v\in\mathbb{R}^d,\ w=\norm{u}^2=\norm{v}^2 \nonumber\\
		&\ c_{\alpha(\ind{i})}=\langle u,A_1(i_1)\dots A_{2t}(i_{2t})v\rangle\ \mathrm{for\ all}\ \ind{i}\in[n+1]^{2t}\nonumber,
	\end{align}
	where $c_{\alpha}$ are the coefficients of $r$.\footnote{Following \cite{gribling2019semidefinite} one obtains \cref{theo:SDP}, but with the $A_s$ being unitary-valued maps. Every contraction-valued map can be turned into an equivalent unitary-valued map by 
	block-encoding contractions into the top-left corner of unitaries.} 
\end{theorem}

Second, we show that in the case of $p$ belonging to $W_\mathcal{P}$, we can restrict $r$ to belong to $W_\mathcal{P}$. As before, we also show that a single contraction-valued map $A$ suffices.\footnote{In fact, all of the contraction-valued maps $A_s$ can be taken to be the same regardless of whether $p$ belongs to $W_\mathcal{P}$ or not, similarly to what is done in Proposition \ref{prop:polarizationTensors}.}
 
\begin{lemma}\label{lem:SDPDualforblockmulti}
	Let $\mathcal{P}$ be a partition of $[n]$ in $2t$ subsets and let $p\in W_\mathcal{P}$. Then,
	\begin{align}\label{eq:SDP2}
		\err(p,t)=\sup&\ (\langle p,r\rangle-w)/\norm{r}_1\\
		\mathrm{s.t.}&\ \nonumber r:\{-1,1\}^n\to\mathbb{R},\, r\in W_\mathcal{P},\, d\in\mathbb{R}\\
		&\ A\colon [n+1]\to B_{M(d)},\, \nonumber\\
		&\ u,v\in\mathbb{R}^d,\, w=\norm{u}^2=\norm{v}^2\nonumber \\
		&\ c_{\alpha(\ind{i})}=\langle u,A(\ind{i})v\rangle\ \mathrm{for\ all}\ \ind{i}\in[n+1]^{2t}\nonumber\ ,
	\end{align}
	where $c_{\alpha}$ are the coefficients of $r$. 
\end{lemma}
\begin{proof}
	Let $\mathcal{E}^*(p,t)$ be the expression in the right-hand side of \cref{eq:SDP2}. By inclusion of the feasible region,  $\err (p,t)\geq \err^*(p,t)$. To prove the other inequality, consider a feasible instance $(u,v,w,A_s,r)$ for the SDP \eqref{eq:SDP}. Consider the contraction-valued map $\tilde{A}(i):=\sum_{s\in[2t]}e_se_{s+1}^\mathsf{T}\otimes \hat{A}_s(i)$ for $i\in [n]$, where $\hat{A}_s(i)$ is a $d2^{2t}\times d2^{2t}$ matrix
	defined by $$\hat{A}_s(i)=\bigoplus_{z\in \{-1,1\}^{2t}} (A_s\cdot z)(i).$$ We also define $\tilde{A}(n+1):=0$ and the vectors $\tilde{u}=e_1\otimes \hat{u}$ and $\tilde{v}=e_{2t+1}\otimes \hat{v}$, where $\hat{u}$ is the $2^td$-dimensional vector defined as the (normalized) direct sum of $2^t$ copies of $u$, i.e., $$\hat{u}=\frac{1}{\sqrt{2^{2t}}}\bigoplus_{z\in\{-1,1\}^{2t}}u,$$ and the same for $\hat{v}$, but with an appropriate sign in each of the copies $$\hat{v}=\frac{1}{\sqrt{2^{2t}}}\bigoplus_{z\in\{-1,1\}^{2t}}v\prod_{I \in \mathcal P}z_I.$$  This way,  $(\tilde{u},\tilde{v},w,\tilde{A},\Pi_\mathcal{P}r)$ is a feasible instance of $\mathcal{E}^*(p,t).$ Indeed, if $\ind{i}$ takes the value $n+1$ at least once or has any repeated indices, then $\alpha(\ind{i})_1+\dots+\alpha(\ind{i})_n<2t$, so $(\Pi_\mathcal{P}r)_{\alpha(\ind{i})}=0$ because $\Pi_\mathcal{P}r\in W_\mathcal{P}$, and $\langle\tilde{u},\tilde{A}(\ind{i})\tilde{v}\rangle=0$ by construction. If~$\ind{i}\in [n]^{2t}$ and has no repeated indices, we get that 
	\[
	\langle \tilde{u},\tilde{A}(\ind{i})\tilde{v}\rangle=\frac{1}{2^{2t}}\sum_{z\in\{-1,1\}^{2t}}\langle u,A_1(i_1)\dots A_{2t}(i_{2t})v\rangle \prod_{I \in \mathcal P}z_I^{1+\sum_{j \in I} \alpha(\ind i)_j}. 
	\]
	Now, reasoning as in \cref{prop:projections}, we get that $$\langle \tilde{u},\tilde{A}(\ind{i})\tilde{v}\rangle=\left\{\begin{array}{ll}
	\langle u,A_1(i_1)\dots A_{2t}(i_{2t})v\rangle & \mathrm{if}\ \ind{i}\ \text{takes  one  value  in  each }I \in \mathcal P,\\
	0 & \mathrm{otherwise},
\end{array}\right.$$
so putting everything together we get that $$\tilde{c}_{\alpha(\ind{i})}=\langle \tilde{u},\tilde{A}(\ind{i})\tilde{v}\rangle, $$ where $\tilde{c}_{\alpha}$ are the coefficients of $\Pi_\mathcal{P}r$. Since $\tilde A$ is contraction-valued and both 
\begin{equation*}
\norm{\tilde u}^2=\frac{1}{2^{2t}}\sum_{z\in\{-1,1\}^{2t}}\norm{u}^2=\norm{u}^2=w    
\end{equation*}
and $\norm{\tilde{v}}^2=w$, we conclude that $(\tilde{u},\tilde{v},w,\tilde{A},\Pi_\mathcal{P}r)$ is a feasible instance of \cref{eq:SDP2}.

	Finally, the value of $(\tilde{u},\tilde{v},w,\tilde{A},\Pi_\mathcal{P}r)$ is at least as large as the one of 
	$(u,v,w,A,r)$: 
	$$\frac{\langle p,\Pi_\mathcal{P}r\rangle-w}{\norm{\Pi_\mathcal{P}r}_1}=\frac{\langle p,r\rangle-w}{\norm{\Pi_\mathcal{P}r}_1}\geq \frac{\langle p,r\rangle-w}{\norm{r}_1},$$
	where in the equality we have used that that $p$  belongs to $W_\mathcal{P}$ and in the inequality we have used \cref{cor:projections}.
\end{proof}

Now, we are ready to prove \cref{theo:lowerSDP}, again.

\begin{proof}[ of \cref{theo:lowerSDP}]
	First note that given a feasible instance $(u,v,w,A,r)$ for \cref{eq:SDP2} we clearly have that $r_{=2t}\in V_{\mathcal P}$. We show that also $w\geq \norm{r_{=2t}}_{\cb,*}$.  By \cref{prop:Tcb*,prop:dualitycb}, this requires us to show that $t!$ times the unique symmetric $2t$-tensor $T_{r_{=2t}}$ associated to $r_{=2t}$ belongs to $w K(n,t)$. To do so, we show that $t! T_{r_{=2t}} = (\langle u, A(\ind i) v\rangle)_{\ind i \in [n]^{2t}}$.  
	Let $\ind{i}\in [n]^{2t}$. If $\ind{i}$ has repeated indices then $(T_{r_{=2t}})_{\ind{i}}=0$ because $r_{=2t}$ is multilinear, and also $\langle u,A(\ind{i})v\rangle=0$, because $r$ is a feasible instance of \cref{eq:SDP2}. If $\ind{i}$ does not have repeated indices, then $t!(T_{r_{=2t}})_{\ind{i}}=c_{\alpha(\ind{i})}$, and also $\langle u,A(\ind{i})v\rangle=c_{\alpha(\ind{i})}$ because $r$ is a feasible solution for Eq. \eqref{eq:SDP2}. 
	
	On the other hand, given $r\in W_\mathcal{P}$ there is an instance $(u,v,w,A,r)$ with $w=\norm{r_{=2t}}_{\cb,*}$. Indeed, by \cref{prop:dualitycb} there is a map $A\colon [n]\to B_{M(d)}$ and vectors $u,v$ whose norm squared is $\norm{r_{=2t}}_{\cb,*}$ such that $c_{\alpha}(r_{=2t})=\langle u,A(\ind{i})v\rangle$ for every $\ind{i}\in\mathcal{I}_\alpha$ and every $\alpha\in\mathbb{Z}_{\geq 0}^n$. Note that in order to have a feasible instance for \cref{eq:SDP2} we need to satisfy its last condition, and with these contractions we can only satisfy it for $\alpha$ such that $\alpha_1+\dots+\alpha_n=2t$. To satisfy it for every $\alpha$ with $\alpha_1+\dots+\alpha_n\leq 2t$, we just have to change the contractions and the  vectors by  $\hat{A}(i):=\sum_{s\in[2t]}e_se_{s+1}^\mathsf{T}\otimes A_s(i)$ and   $\hat{u}:=e_1\otimes u$ and $\hat{v}:=e_{t+1}\otimes v$ and define the extra contraction as $\hat{A}(n+1)=0$. This, way $(\hat{u},\hat{v},\norm{r_{=2t}}_{\cb,*},\hat{A},r)$ is a feasible instance. To sum up, so far we have proved that 
	\begin{align*}
		\err(p,t)=\sup&\ (\langle p,r\rangle-\norm{r_{=2t}}_{\cb,*})/\norm{r}_1\\
		\mathrm{s.t.}&\ r:\{-1,1\}^n\to\mathbb{R},\ r\in W_\mathcal{P}.
	\end{align*}
	We finally reformulate the above in terms of $r_{=2t}$ using the following two observations. Since $p \in V_{\mathcal P}$ we have $\langle p,r\rangle=\langle p,r_{=2t}\rangle$. Moreover, by \cref{prop:dualityinfty}, we have $\norm{r_{=2t}}_{\infty,*} = \inf\{ \norm{\tilde r}_1 \st \tilde r:\{-1,1\}^n\to\mathbb{R},\, \tilde r\in W_\mathcal{P},\, \tilde r_{=2t}=r_{=2t}\}$. Hence, 
		\begin{align*}
		\err(p,t)=\sup&\ (\langle p,r_{=2t} \rangle-\norm{r_{=2t}}_{\cb,*})/\norm{r_{=2t}}_{\infty,*}\\
		\mathrm{s.t.}&\ r_{=2t}\in V_\mathcal{P},
	\end{align*}
	which concludes the proof.
\end{proof}

\subsection*{Acknowledgments}

We thank anonymous referees for their helpful feedback. 

\bibliographystyle{alphaurl}
\bibliography{Bibliography}

\newcommand{\etalchar}[1]{$^{#1}$}
\begin{thebibliography}{MMFPSS22}

\bibitem[AAI{\etalchar{+}}16]{Aaronson2015PolynomialsQQ}
Scott Aaronson, Andris Ambainis, J\=anis Iraids, Martins Kokainis, and Juris
  Smotrovs.
\newblock Polynomials, quantum query complexity, and {G}rothendieck's
  inequality.
\newblock In {\em 31st Conference on Computational Complexity, {CCC} 2016},
  pages 25:1--25:19, 2016.
\newblock arXiv:1511.08682.
\newblock URL: \url{https://doi.org/10.4230/LIPIcs.CCC.2016.25}.

\bibitem[Aar21]{Aaronson:2021}
Scott Aaronson.
\newblock Open problems related to quantum query complexity.
\newblock {\em ACM Transactions on Quantum Computing}, 2(4), 2021.
\newblock URL: \url{https://doi.org/10.1145/3488559}.

\bibitem[AB23]{ambainis2023exponentialccc}
Andris Ambainis and Aleksandrs Belovs.
\newblock An exponential separation between quantum query complexity and the
  polynomial degree.
\newblock In {\em Proceedings of the Conference on Proceedings of the 38th
  Computational Complexity Conference}, CCC '23, Dagstuhl, DEU, 2023. Schloss
  Dagstuhl--Leibniz-Zentrum fuer Informatik.
\newblock URL: \url{https://doi.org/10.4230/LIPIcs.CCC.2023.24}.

\bibitem[ABDK16]{AaronsonBDK:2016}
Scott Aaronson, Shalev Ben-David, and Robin Kothari.
\newblock Separations in query complexity using cheat sheets.
\newblock STOC '16, New York, NY, USA, 2016. Association for Computing
  Machinery.
\newblock URL: \url{https://doi.org/10.1145/2897518.2897644}.

\bibitem[ABP19]{QQA=CBF}
Srinivasan Arunachalam, Jop Bri{\"{e}}t, and Carlos Palazuelos.
\newblock Quantum query algorithms are completely bounded forms.
\newblock {\em SIAM J.\ Comput}, 48(3):903--925, 2019.
\newblock Preliminary version in ITCS'18.
\newblock URL: \url{https://doi.org/10.1137/18M117563X}.

\bibitem[Amb06]{Ambainis:2006}
A.~Ambainis.
\newblock Polynomial degree vs.\ quantum query complexity.
\newblock {\em J. Comput. System Sci.}, 72(2):220--238, 2006.
\newblock Earlier version in FOCS'03. quant-ph/0305028.
\newblock URL: \url{https://doi.org/10.1016/j.jcss.2005.06.006}.

\bibitem[Amb07]{Ambainis:2007}
A.~Ambainis.
\newblock Quantum walk algorithm for element distinctness.
\newblock {\em SIAM Journal on Computing}, 37(1):210--239, 2007.
\newblock Earlier version in FOCS'04. arXiv:quant-ph/0311001.
\newblock URL: \url{https://doi.org/10.1137/S0097539705447311}.

\bibitem[Amb18]{Ambainis:2018}
Andris Ambainis.
\newblock Understanding quantum algorithms via query complexity.
\newblock In {\em Proceedings of the International Congress of Mathematicians
  (ICM 2018)}, pages 3265--3285, 2018.
\newblock \href {https://doi.org/10.1142/9789813272880_0181}
  {\path{doi:10.1142/9789813272880_0181}}.

\bibitem[BBB{\etalchar{+}}19]{BBBLL:2019}
Tom Bannink, Jop Bri{\"e}t, Harry Buhrman, Farrokh Labib, and Troy Lee.
\newblock {Bounding Quantum-Classical Separations for Classes of Nonlocal
  Games}.
\newblock In {\em 36th International Symposium on Theoretical Aspects of
  Computer Science (STACS 2019)}, volume 126 of {\em Leibniz International
  Proceedings in Informatics (LIPIcs)}, pages 12:1--12:11, Dagstuhl, Germany,
  2019. Schloss Dagstuhl--Leibniz-Zentrum fuer Informatik.
\newblock Available at arXiv: 1811.11068.
\newblock URL: \url{http://doi.org/10.4230/LIPIcs.STACS.2019.12}.

\bibitem[BBC{\etalchar{+}}01]{polynomialmethod}
Robert Beals, Harry Buhrman, Richard Cleve, Michele Mosca, and Ronald de~Wolf.
\newblock Quantum lower bounds by polynomials.
\newblock {\em J. ACM}, 48(4):778–797, 2001.
\newblock URL: \url{https://doi.org/10.1145/502090.502097}.

\bibitem[BE22]{briet2022converses}
Jop Bri\"{e}t and Francisco Escudero{ }Guti\'{e}rrez.
\newblock {On Converses to the Polynomial Method}.
\newblock In {\em 17th Conference on the Theory of Quantum Computation,
  Communication and Cryptography (TQC 2022)}, volume 232 of {\em Leibniz
  International Proceedings in Informatics (LIPIcs)}, pages 6:1--6:10. Schloss
  Dagstuhl -- Leibniz-Zentrum f{\"u}r Informatik, 2022.
\newblock URL: \url{http://doi.org/10.4230/LIPIcs.TQC.2022.6}.

\bibitem[BH31]{bohnenblust1931absolute}
Henri~Fr{\'e}d{\'e}ric Bohnenblust and Einar Hille.
\newblock On the absolute convergence of dirichlet series.
\newblock {\em Annals of Mathematics}, pages 600--622, 1931.
\newblock URL: \url{https://doi.org/10.2307/1968255}.

\bibitem[BKT20]{BunKhotariThaler:2020}
Mark Bun, Robin Kothari, and Justin Thaler.
\newblock The polynomial method strikes back: Tight quantum query bounds via
  dual polynomials.
\newblock {\em Theory of Computing}, 16(10):1--71, 2020.
\newblock URL: \url{https://doi.org/10.4086/toc.2020.v016a010}.

\bibitem[BLM13]{boucheron:concentration}
S.~Boucheron, G.~Lugosi, and P.~Massart.
\newblock {\em Concentration inequalities: A nonasymptotic theory of
  independence}.
\newblock Oxford university press, 2013.
\newblock URL: \url{https://doi.org/10.1093/acprof:oso/9780199535255.001.0001}.

\bibitem[BMMN13]{Braverman:2013}
M.~Braverman, K.~Makarychev, Y.~Makarychev, and A.~Naor.
\newblock The {G}rothendieck constant is strictly smaller than {K}rivine's
  bound.
\newblock {\em Forum Math. Pi}, 1:453--462, 2013.
\newblock Preliminary version in FOCS'11. arXiv:1103.6161.
\newblock URL: \url{https://doi.org/10.1017/fmp.2013.4}.

\bibitem[BP19]{briet2018failure}
Jop Bri\"{e}t and Carlos Palazuelos.
\newblock Failure of the trilinear operator space {G}rothendieck inequality.
\newblock {\em Discrete Analysis}, 2019.
\newblock Paper No.~8.
\newblock URL: \url{https://doi.org/10.19086/da.8805}.

\bibitem[BSdW22]{Bansal:2022}
Nikhil Bansal, Makrand Sinha, and Ronald de~Wolf.
\newblock {Influence in Completely Bounded Block-Multilinear Forms and
  Classical Simulation of Quantum Algorithms}.
\newblock In {\em 37th Computational Complexity Conference (CCC 2022)}, volume
  234, pages 28:1--28:21, 2022.
\newblock URL: \url{http://doi.org/10.4230/LIPIcs.CCC.2022.28}.

\bibitem[BV04]{boyd_vandenberghe_2004}
Stephen Boyd and Lieven Vandenberghe.
\newblock {\em Convex Optimization}.
\newblock Cambridge University Press, 2004.
\newblock URL: \url{http://doi.org/10.1017/CBO9780511804441}.

\bibitem[Dav84]{Davie:1984}
A.~Davie.
\newblock Lower bound for {$K_G$}.
\newblock Unpublished, 1984.

\bibitem[Esc24]{gutierrez2023influences}
Francisco Escudero{\ }Guti{\'e}rrez.
\newblock Influences of fourier completely bounded polynomials and classical
  simulation of quantum algorithms.
\newblock {\em Chicago Journal of Theoretical Computer Science}, 2024.
\newblock URL: \url{https://doi.org/10.48550/arXiv.2304.06713}.

\bibitem[GL19]{gribling2019semidefinite}
Sander Gribling and Monique Laurent.
\newblock Semidefinite programming formulations for the completely bounded norm
  of a tensor.
\newblock 2019.
\newblock URL: \url{https://doi.org/10.48550/arXiv.1901.04921}.

\bibitem[Gre07]{Green:2007}
Ben Green.
\newblock Montr\'{e}al notes on quadratic {F}ourier analysis.
\newblock In {\em Additive combinatorics}, volume~43 of {\em CRM Proc. Lecture
  Notes}, pages 69--102. Amer. Math. Soc., Providence, RI, 2007.
\newblock URL: \url{https://doi.org/10.1090/crmp/043/06}.

\bibitem[Gro53]{grothendieck1953resume}
Alexandre Grothendieck.
\newblock {\em R{\'e}sum{\'e} de la th{\'e}orie m{\'e}trique des produits
  tensoriels topologiques}.
\newblock Soc. de Matem{\'a}tica de S{\~a}o Paulo, 1953.
\newblock URL: \url{http://doi.org/10.5802/aif.46}.

\bibitem[Gro96]{Grover:1996}
L.~K. Grover.
\newblock A fast quantum mechanical algorithm for database search.
\newblock In {\em Proceedings of the twenty-eighth annual ACM symposium on
  Theory of computing}, pages 212--219. ACM, 1996.
\newblock URL: \url{https://doi.org/10.1145/237814.237866}.

\bibitem[Har72]{harris1972bounds}
Lawrence~A Harris.
\newblock Bounds on the derivatives of holomorphic functions of vectors.
\newblock In {\em Proc. Colloq. Analysis, Rio de Janeiro}, volume 145, page
  163, 1972.

\bibitem[HW{\etalchar{+}}79]{hardy1979introduction}
Godfrey~Harold Hardy, Edward~Maitland Wright, et~al.
\newblock {\em An introduction to the theory of numbers}.
\newblock Oxford university press, 1979.
\newblock URL: \url{https://doi.org/10.1126/science.90.2329.158.b}.

\bibitem[KM13]{kane2013prg}
Daniel~M Kane and Raghu Meka.
\newblock A {PRG} for {L}ipschitz functions of polynomials with applications to
  sparsest cut.
\newblock In {\em Proceedings of the forty-fifth annual ACM symposium on Theory
  of computing}, pages 1--10, 2013.
\newblock URL: \url{https://doi.org/10.1145/2488608.2488610}.

\bibitem[KN07]{khot2007linear}
Subhash Khot and Assaf Naor.
\newblock Linear equations modulo 2 and the l1 diameter of convex bodies.
\newblock In {\em 48th Annual IEEE Symposium on Foundations of Computer Science
  (FOCS'07)}, pages 318--328. IEEE, 2007.
\newblock URL: \url{https://doi.org/10.1109/FOCS.2007.20}.

\bibitem[Lov10]{lovett2010elementary}
Shachar Lovett.
\newblock An elementary proof of anti-concentration of polynomials in gaussian
  variables.
\newblock In {\em Electron. Colloquium Comput. Complex.}, volume~17, page 182,
  2010.

\bibitem[MMFPSS22]{moslehian2022similarities}
Mohammad~Sal Moslehian, GA~Mu{\~n}oz-Fern{\'a}ndez, AM~Peralta, and
  JB~Seoane-Sep{\'u}lveda.
\newblock Similarities and differences between real and complex banach spaces:
  an overview and recent developments.
\newblock {\em Revista de la Real Academia de Ciencias Exactas, F{\'\i}sicas y
  Naturales. Serie A. Matem{\'a}ticas}, 116(2):1--80, 2022.
\newblock URL: \url{https://doi.org/10.1007/s13398-022-01222-8}.

\bibitem[Nik54]{Nikaid1954OnVN}
Hukukane Nikaid{\^o}.
\newblock {On von Neumann’s minimax theorem}.
\newblock {\em Pacific Journal of Mathematics}, 4:65--72, 1954.

\bibitem[O'D14]{o2014analysis}
Ryan O'Donnell.
\newblock {\em Analysis of boolean functions}.
\newblock Cambridge University Press, 2014.
\newblock URL: \url{https://doi.org/10.1017/CBO9781139814782}.

\bibitem[OZ15]{o2015polynomial}
Ryan O'Donnell and Yu~Zhao.
\newblock Polynomial bounds for decoupling, with applications.
\newblock {\em arXiv preprint arXiv:1512.01603}, 2015.
\newblock URL: \url{https://doi.org/10.4230/LIPIcs.CCC.2016.24}.

\bibitem[Pau03]{paulsenoperatoralgebras}
Vern Paulsen.
\newblock {\em Completely Bounded Maps and Operator Algebras}.
\newblock 02 2003.
\newblock URL: \url{https://doi.org/10.1017/CBO9780511546631}.

\bibitem[Ree91]{Reeds:1991}
J.~Reeds.
\newblock A new lower bound on the real {Grothendieck} constant.
\newblock Manuscript (\url{http://www.dtc.umn.edu/~reedsj/bound2.dvi}), 1991.

\bibitem[Sho97]{Shor:1997}
P.~W. Shor.
\newblock Polynomial-time algorithms for prime factorization and discrete
  logarithms on a quantum computer.
\newblock {\em {SIAM} Journal of Computing}, 26(5):1484--1509, 1997.
\newblock Earlier version in FOCS'94.
\newblock URL: \url{https://doi.org/10.1137/S0097539795293172}.

\bibitem[Tao12]{tao2012topics}
Terence Tao.
\newblock {\em Topics in Random Matrix Theory}.
\newblock Graduate studies in mathematics. American Mathematical Society, 2012.

\bibitem[Tsi80]{Tsirelson}
B.~S. Tsirelson.
\newblock Quantum generalizations of {B}ell's inequality.
\newblock {\em Letters in Mathematical Physics}, 1980.
\newblock URL: \url{https://doi.org/10.1007/BF00417500}.

\bibitem[TT23]{tao2021quantitativej}
Terence Tao and Joni Ter{\"a}v{\"a}inen.
\newblock Quantitative bounds for {G}owers uniformity of the {M}\"{o}bius and
  von {M}angoldt functions.
\newblock {\em Journal of the European Mathematical Society}, 2023.
\newblock \href {https://doi.org/10.4171/jems/1404}
  {\path{doi:10.4171/jems/1404}}.

\bibitem[Var74]{Varopoulos:1974}
N.~Th. Varopoulos.
\newblock On an inequality of von {N}eumann and an application of the metric
  theory of tensor products to operators theory.
\newblock {\em J. Functional Analysis}, 16:83--100, 1974.
\newblock URL: \url{http://doi.org/10.1016/0022-1236(74)90071-8}.

\end{thebibliography}

\end{document}